\documentclass{icmart}

\contact[info@boazbarak.org]{Boaz Barak}
\contact[dsteurer@cs.cornell.edu]{David Steurer}

\usepackage{mathtools}
\usepackage{nicefrac}
\usepackage{url}
\usepackage{dsfont}
\usepackage[normalem]{ulem}
\usepackage{xcolor}
\usepackage{marginnote} 
\let\mathbb\mathds

\newtheorem{theorem}{Theorem}[section]
\newtheorem{corollary}[theorem]{Corollary}
\newtheorem{lemma}[theorem]{Lemma}

\newtheorem{conjecture}[theorem]{Conjecture}

\theoremstyle{definition}
\newtheorem{definition}[theorem]{Definition}
\newtheorem{remark}[theorem]{Remark}

\DeclareMathOperator{\supp}{\mathrm{supp}}
\DeclareMathOperator{\bbE}{\mathbb{E}}
\DeclareMathOperator{\Span}{\mathrm{Span}}

\DeclareMathOperator{\pE}{\Tilde{\mathbb{E}}}

\title[Sum-of-squares proofs and the quest toward optimal algorithms]{Sum-of-Squares Proofs and the\\ Quest toward Optimal Algorithms}

\author[Boaz Barak and David Steurer]{Boaz Barak and David Steurer}

\begin{document}

\begin{abstract}
In order to obtain the best-known guarantees, algorithms are traditionally tailored to the particular problem we want to solve.
Two recent developments, the \emph{Unique Games Conjecture (UGC)} and the \emph{Sum-of-Squares (SOS) method}, surprisingly suggest that this tailoring is not necessary and that a single efficient algorithm could achieve best possible guarantees for a wide range of different problems.

The \emph{Unique Games Conjecture (UGC)} is a tantalizing conjecture in computational complexity, which, if true, will shed light on the complexity of a great many problems.
In particular this conjecture predicts that a \emph{single concrete algorithm} provides optimal guarantees among all efficient algorithms for a large class of computational problems.

The \emph{Sum-of-Squares (SOS) method} is a general approach for solving systems of polynomial constraints.
This approach is studied in several scientific disciplines, including real algebraic geometry, proof complexity, control theory, and mathematical programming, and has found applications in fields as diverse as quantum information theory, formal verification, game theory and many others.

We survey some connections that were recently uncovered between the Unique Games Conjecture and the Sum-of-Squares method.
In particular, we discuss new tools to rigorously bound the running time of the SOS method for obtaining approximate solutions to hard optimization problems,
and how these tools give the potential for the sum-of-squares method to provide new guarantees for many problems of interest, and possibly to even refute the UGC.
\end{abstract}




\maketitle

\section{Introduction}

A central mission of theoretical computer science is to understand which computational problems can be solved efficiently, which ones cannot,
and what it is about a problem that makes it easy or hard.
To illustrate these kind of questions,
let us consider the following parameters of an undirected $d$-regular graph\footnotemark\  $G=(V,E)$:

\footnotetext{
  An undirected $d$-regular graph $G=(V,E)$ consists of a set of \emph{vertices} $V$, which we sometimes identify with the set $[n]=\{1,\ldots,n\}$ for some integer $n$, and a set of \emph{edges} $E$, which are $2$-element subsets of $V$, such that
  every vertex is part of exactly $d$ edges.
  The assumption that $G$ is regular is not important and made chiefly for notational simplicity.
  For vertex sets $S,T \subseteq V$, we let $E(S,T)$ denote the set of edges $ \{s,t\}\in E$ with $ s\in S$ and $t\in T$.
}

\begin{itemize}

\item The \emph{smallest connected component of $G$}  is the size of the smallest non-empty set $S\subseteq V$ such that
$E(S,V\setminus S) = \emptyset$.

\item The \emph{independent-set number of $G$} is the size of the largest set $S\subseteq V$ such that $E(S,S) = \emptyset$.

\item The \emph{(edge) expansion\footnotemark\ of $G$}, denoted $\phi_G$, is the minimum \emph{expansion} $\phi_G(S)$ of a vertex set $S\subseteq V$ with size $1\le \lvert  S \rvert\le \lvert  V \rvert/2$, where
\begin{displaymath}
  \phi_G(S)= \frac{|E(S,V\setminus S)|}{d|S|}
  \,.
\end{displaymath}
The expansion $\phi_G(S)$ measures the probability that a step of the random walk on $G$ leaves $S$ conditioned on starting in $S$.
\end{itemize}

\footnotetext{
  The expansion of a graph is closely related to other quantities, known as \emph{isoperimetric constant}, \emph{conductance} or \emph{sparsest cut}.
  These quantities are not identical but are the same up to scaling and a multiplicative factor of at most $2$.
  Hence, they are computationally equivalent for our purposes.
  We also remark that expansion is often not normalized by the degree.
  However for our purposes this normalization is useful.
}

All these parameters capture different notions of well-connectedness of the graph $G$.
Computing these can be very useful in many of the settings in which we use graphs to model data, whether it is communication links between servers, social connections between people, genes that are co-expressed together, or transitions between states of a system.

The computational complexity of the first two parameters is fairly well understood.
The smallest connected component is easy to compute in time linear in the number $n=|V|$ of vertices  by using, for example, breadth-first search from every vertex in the graph.
The independent-set number is \ensuremath{\mathbf{NP}}-hard to compute, which means that, assuming the widely believed conjecture that $\ensuremath{\mathbf{P}}\neq\ensuremath{\mathbf{NP}}$, it cannot be computed in time polynomial in $n$.
In fact, under stronger (but still widely believed) quantitative versions of the $\ensuremath{\mathbf{P}}\neq \ensuremath{\mathbf{NP}}$ conjecture, for every $k$ it is infeasible to decide whether or not the maximum independent set is larger than $k$ in time $n^{o(k)}$~\cite{downey1995fixed,chen2006strong}
and hence we cannot significantly beat the trivial $O(n^k)$-time algorithm for this problem.
Similarly, while we can approximate the independent-set number trivially within a factor of $n$, assuming such conjectures, there is no polynomial-time algorithm to approximate it within a factor of $n^{1-\varepsilon(n)}$ where $\varepsilon(n)$ is some  function tending to zero as $n$ grows~\cite{Hastad96,Khot01}.

So, connectivity is an easy problem and independent set a hard one, but what about expansion? Here the situation is more complicated.
We know that we can't efficiently compute $\phi_G$ exactly, and we can't even get an arbitrarily good approximation~\cite{AmbuhlMS11}, but we actually do have efficient algorithms with non-trivial approximation guarantees for $\phi_G$.
Discrete versions of \emph{Cheeger's inequality}~\cite{cheeger1970lower,dodziuk1984difference,AlonM85,Alon86} yield such an estimate, namely
\begin{equation}
  \tfrac{d- \lambda_2}{2d} \leq \phi_G \leq 2\sqrt{\tfrac{d-\lambda_2}{2d}}
  \,,
  \label{eq:cheeger}
\end{equation}
where $\lambda_2(G)$ denotes the (efficiently computable) second largest eigenvalue of the $G$'s adjacency matrix.\footnote{The \emph{adjacency matrix} of a graph $G$ is the $|V|\times |V|$ matrix $A$ with $0/1$ entries such that $A_{u,v}=1$ iff $\{u,v\}\in E$.}
In particular, we can use (\ref{eq:cheeger}) to efficiently distinguish between graphs with $\phi_G$ close to $0$ and graphs with $\phi_G$ bounded away from $0$.
But can we do better?
For example, could we efficiently compute a quantity $c_G$ such that  $c_G\leq \phi_G \leq O(c_G^{0.51})$?
We simply don't know.\footnote{As we will mention later, there are algorithms to approximate $\phi_G$ up to factors depending on the number $n$ of vertices, which give better guarantees than (\ref{eq:cheeger}) for graphs where $\phi_G$ is sufficiently small as a function of $n$.}

\medskip

This is not an isolated example, but a pattern that keeps repeating.
Over the years, computer scientists have developed sophisticated tools to come up with algorithms on one hand, and hardness proofs showing the limits of efficient algorithms on the other hand.
But those two rarely match up.
Moreover, the cases where we do have tight hardness results are typically in settings, such as the independent set problem, where there is no way to significantly beat the trivial algorithm.
In contrast, {} for problems such as computing expansion, where we already know of an algorithm giving non-trivial guarantees, we typically have no proof that this algorithm is
\emph{optimal}.
In other words, the following is a common theme:

\begin{quote}
\emph{If you already know an algorithm with non-trivial approximation guarantees for a problem, it's very hard to rule out that cleverer algorithms couldn't get even better guarantees.}
\end{quote}

\noindent
In 2002, Subhash Khot formulated a conjecture, known as the \emph{Unique Games Conjecture (UGC)}~\cite{Khot02}.
A large body of follow up works has shown that this conjecture (whose description is deferred to Section~\ref{sec:UGC-SSEH} below) implies many hardness results that overcome the above challenge and match the best-known algorithms even in cases when they achieve non-trivial guarantees.
In fact, beyond just resolving particular questions, this line of works obtained far-reaching complementary \emph{meta algorithmic} and \emph{meta hardness} results.
By this we mean results that give an efficient \emph{meta algorithm} $\mathcal{A}$ (i.e., an algorithm that can be applied to a family of problems, and not just a single one) that is
\emph{optimal} within a broad domain $\mathcal{C}$, in the sense that (assuming the UGC) there is no polynomial-time algorithm that performs better than $\mathcal{A}$ on any problem in $\mathcal{C}$.
It is this aspect of the Unique Games Conjecture result that we find most exciting, and that shows promise of going beyond the current state where the individual algorithmic and hardness results form ``isolated islands of knowledge surrounded by a sea of ignorance''\footnote{{}Paraphrasing John Wheeler.} into a more unified theory of complexity.

The meta-algorithm that the UGC predicts to be optimal is based on \emph{semidefinite programming} and it uses this technique in a very particular and quite restricted way.
(In many settings, this meta-algorithm can be implemented in near-linear time~\cite{Steurer10}.)
We will refer to this algorithm as the \emph{UGC meta-algorithm}.
It can be viewed as a common generalization of several well known algorithms,
including those that underlie Cheeger's Inequality, Grothendieck's Inequality~\cite{grothendieck1953resume}, the Goemans--Williamson \ensuremath{\text{\textsc{Max Cut}}}\ algorithm~\cite{GoemansW95}, and the Lov\'{a}sz\ $\vartheta$ function~\cite{lovasz1979shannon}.
As we've seen for the example of Cheeger's Inequality, in many of those settings this meta-algorithm gives \emph{non-trivial approximation guarantees} which are the best known, but there are no hardness results ruling out the existence of better algorithms.
The works on the UGC has shown that this conjecture (and related ones) imply that this meta-algorithm is \emph{optimal} for a vast number of problems, including all those examples above.
For example, a beautiful result of Raghavendra~\cite{Raghavendra08} showed that for every constraint-satisfaction problem (a large class of problems
that includes many problems of interest such as \textsc{Max $k$-SAT}, \textsc{$k$-Coloring}, and \textsc{Max-Cut}),
the UGC meta-algorithm gives the best estimate on the maximum possible fraction of constraints one can satisfy.
Similarly, the UGC (or closely related variants) imply there are no efficient algorithms that give a better estimate for the sparsest cut of a graph than the one implied by Cheeger's Inequality~\cite{RaghavendraST12} and no better efficient estimate for the maximum correlation of a matrix with $\pm 1$-valued vectors than the one given by Grothendieck's Inequality.\footnote{See~\cite{RaghavendraS09groth} for the precise statement of Grothendieck's Inequality and this result.
Curiously, the UGC implies that Grothendieck's Inequality yields the best efficient approximation factor for the correlation of a matrix with $\pm 1$-valued vectors even though we don't actually know the numerical value of this factor (known as Grothendieck's constant).}
To summarize:

\begin{quote}
\emph{If true, the Unique Games Conjecture tells us not only \emph{which} problems in a large class are easy and which are hard, but also \emph{why} this is the case.
There is a \emph{single unifying reason}, captured by a concrete meta-algorithm, that explains all the easy problem in this class.
Moreover, in many cases where this meta-algorithm already gives non-trivial guarantees, the UGC implies that no further efficient improvements are possible.}
\end{quote}

All this means that the Unique Games Conjecture is certainly a very attractive proposition, but the big question still remains unanswered---is this conjecture actually true?
While some initial results supported the UGC, more recent works, although still falling short of disproving the conjecture, have called it into question.
In this survey we discuss the most promising current approach to refute the UGC, which is based on the \emph{Sum of Squares (SOS) method}~\cite{Shor87,Nesterov00,Parrilo00,Lasserre01}.
The SOS method could potentially refute the Unique Games Conjecture by beating the guarantees of the UGC meta-algorithm on problems on which the conjecture implies the latter's optimality.
This of course is interesting beyond the UGC, as it means we would be able to improve the known guarantees for many problems of interest.
Alas, analyzing the guarantees of the SOS method is a very challenging problem, and we still have relatively few tools to do so.
However, as we will see, we already know that at least in some contexts, the SOS method can yield better results than what was known before.
The SOS method is itself a meta algorithm, so even if it turns out to refute the UGC, this does not mean we need to give up on the notion of explaining the complexity
of wide swaths of problems via a single algorithm; we may just need to consider a different algorithm.
To summarize, regardless of whether it refutes the UGC or not, understanding the power of the SOS method  is an exciting research direction
that could advance us further towards the goal of a unified understanding of computational complexity.

\subsection{The UGC and SSEH conjectures}
\label{sec:UGC-SSEH}

Instead of the Unique Games Conjecture, in this survey we focus on a related conjecture known as the
\emph{Small-Set Expansion Hypothesis (SSEH)}~\cite{RaghavendraS10}.
The SSEH implies the UGC~\cite{RaghavendraS10}, and while there is no known implication in the other direction,
there are several results suggesting that these two conjectures are probably equivalent~\cite{RaghavendraS10,RaghavendraST10,RaghavendraS09,AroraBS10,BarakBHKSZ12}.
At any rate, most (though not all) of what we say in this survey applies equally well to both conjectures,
but the SSEH is, at least in our minds, a somewhat more natural and simpler-to-state conjecture.

Recall that for a $d$-regular graph $G=(V,E)$ and a vertex set $S\subseteq V$, we defined its expansion as $\phi_G(S)=|E(S,V\setminus S)|/(d|S|)$.
By Cheeger's inequality (\ref{eq:cheeger}), the second largest eigenvalue yields a non-trivial approximation for the minimum expansion $\phi_G = \min_{1\le|S|\leq |V|/2}\phi_G(S)$,
but it turns out that eigenvalues and similar methods do not work well for the problem of approximating the minimum expansion of smaller sets.
The Small-Set Expansion Hypothesis conjectures that this problem is inherently difficult.

\begin{conjecture}[Small-Set Expansion Hypothesis~\cite{RaghavendraS10}] For every $\varepsilon>0$ there exists $\delta>0$
such that given any graph $G=(V,E)$, it is $\ensuremath{\mathbf{NP}}$-hard to distinguish between the case \textbf{(i)} that there exists a subset $S\subseteq V$ with $|S|=\delta |V|$ such that $\phi_G(S) \leq \varepsilon$ and the case \textbf{(ii)} that
$\phi_G(S) \geq 1-\varepsilon$ for every $S$ with $|S| \leq \delta |V|$.
\end{conjecture}

As mentioned above, the SSEH implies that (\ref{eq:cheeger}) yields an optimal approximation for $\phi_G$.
More formally, assuming the SSEH, there is some absolute constant $c>0$ such that for every $\phi \geq 0$, it is $\ensuremath{\mathbf{NP}}$-hard to distinguish between the case that a given graph $G$ satisfies $\phi_G \leq \phi$ and the case that  $\phi_G \geq c\sqrt{\phi}$~\cite{RaghavendraST12}.
Given that the SSEH conjectures the difficulty of approximating expansion, the reader might not be
so impressed that it also implies the optimality of Cheeger's Inequality.
However, we should note that the SSEH merely conjectures that the problem becomes harder as $\delta$ becomes smaller, without postulating
any quantitative relation between $\delta$  and $\varepsilon$,  and
so it is actually surprising (and requires a highly non-trivial proof) that it implies such quantitatively tight bounds.
Even more surprising is that (through its connection with the UGC) the SSEH implies tight hardness result for a host of other problems,
including every constraint satisfaction problem, Grothendieck's problem, and many others, which a priori seem to have nothing to do with graph expansion.

\begin{remark}
While we will stick to the SSEH in this survey, for completeness we present here the
definition of the Unique Games Conjecture.
We will not use this definition in the proceeding and so the reader can feel free to skip this remark.
The UGC can be thought of as a more structured variant of the SSEH where we restrict to graphs and sets that satisfy
some particular properties.
Because we restrict both the graphs and the sets, a priori it is not clear which of these conjectures should be stronger.  However it turns out that the SSEH implies the UGC~\cite{RaghavendraS10}.
It is an open problem whether the two conjectures are equivalent, though the authors personally suspect that this is the case.

We say that an $n$-vertex graph $G=(V,E)$ is \emph{$\delta$-structured} if there is a partition of $V$ into $\delta n$
sets $V_1,\ldots,V_{\delta n}$ each of size $1/\delta$, such that for every $i\neq j$, either $E(V_i,V_j)=\emptyset$ or $E(V_i,V_j)$ is a \emph{matching} (namely for every $u\in V_i$ there is exactly one $v\in V_j$ such that $\{u,v\}\in E$).  We say a set $S\subseteq V$ is \emph{$\delta$-structured} if $|S\cap V_i|=1$ for all $i$ (and so in particular, $|S|=\delta n$).
The Unique Games Conjecture states that for every $\varepsilon>0$ there exists a $\delta>0$ such that it is $\ensuremath{\mathbf{NP}}$ hard, given a $\delta$-structured $G$, to distinguish between the case \textbf{(i)} that there exists a $\delta$-structured $S$ such that $\phi_G(S) \leq \varepsilon$ and the case \textbf{(ii)} that every $\delta$-structured $S$ satisfies $\phi_G(S) \geq 1-\varepsilon$.
The conjecture can also be described in the form of so-called ``two prover one round games'' (hence its name); see Khot's surveys~\cite{Khot10a,Khot10b}.
\end{remark}

\medskip

\subsection{Organization of this survey and further reading}
In the rest of this survey we describe the Sum of squares algorithm, some of its applications, and its relation to the Unique Games and Small-Set Expansion Conjectures.
We start by defining the Sum of Squares algorithm, and how it relates to classical questions such as Hilbert $17^{th}$ problem.
We will demonstrate how the SOS algorithm is used, and its connection to the UGC/SSEH, by presenting Cheeger's Inequality (\ref{eq:cheeger}) as an instance of this algorithm.
The SSEH implies that the SOS algorithm cannot yield better estimates to $\phi_G$ than those obtained by (\ref{eq:cheeger}).
While we do not know yet whether this is true or false, we present two different applications where the SOS does beat prior works--- finding a planted sparse vector in a random subspace, and \emph{sparse coding}--- learning a set of vectors $A$ given samples of random sparse linear combinations
of vectors in $A$.
We then discuss some of the evidence for the UGC/SSEH, how this evidence is challenged by the SOS algorithm
and the relation between the UGC/SSEH and the problem of (approximately) finding sparse vectors in arbitrary (not necessarily random) subspaces.
Much of our discussion is based on the papers~\cite{AroraBS10,BarakGHMRS12,BarakBHKSZ12,BarakKS14,BarakKS15}.
See also~\cite{Barak12,Barak13,Barak14} for informal overviews of some of these issues.

For the reader interested in learning more about the Unique Games Conjecture, there are three excellent surveys on this topic.
Khot's CCC survey~\cite{Khot10b} gives a fairly comprehensive overview of the state of knowledge
on the UGC circa 2010, while  his ICM survey~\cite{Khot10a} focuses on some of the techniques and connections that arose in the works around the UGC.
Trevisan~\cite{Trevisan12} gives a wonderfully accessible introduction to the UGC, using the \textsc{Max-Cut} problem as a running example to
explain in detail the UGC's connection to semidefinite programming.
As a sign of how rapidly research in this area is progressing, this survey is almost entirely disjoint from~\cite{Khot10a,Khot10b,Trevisan12}.
While the former surveys mostly described the implications of the UGC for obtaining very strong hardness and ``meta hardness'' results,
the current manuscript is focused on the question of whether the UGC is actually true,
and more generally understanding the power of the SOS algorithm to go beyond the basic LP and SDP relaxations.

Our description of the SOS algorithm barely scratches the surface of this fascinating topic, which has a great many applications
that have nothing to do with the UGC or even approximation algorithms at large.
The volume~\cite{blekherman2013semidefinite} and the monograph~\cite{Laurent09} are good sources for some of these topics.
The SOS algorithm was developed in slightly different forms by several researchers, including Shor~\cite{Shor87}, Nesterov~\cite{Nesterov00}, Parrilo~\cite{Parrilo00}, and Lasserre~\cite{Lasserre01}.
It can be viewed as a strengthening of other ``meta-algorithms'' proposed by~\cite{SheraliA90,LovaszS91} (also known as linear and semi-definite programming hierarchies).\footnote{See~\cite{Laurent03} for a comparison.}
Our description of the SOS meta algorithm follows Parrilo's, while the description of the dual algorithm
follows Lasserre, although we use the pseudoexpectation notation introduced in~\cite{BarakBHKSZ12} instead of
Lasserre's notion of ``moment matrices''.
The Positivstellensatz/SOS proof system was first studied by Grigoriev and Vorobjov~\cite{GrigorievV01} and Grigoriev~\cite{Grigoriev01} proved some degree lower bounds for it, that were later rediscovered and expanded upon by~\cite{Schoenebeck08,Tulsiani09}.
All these are motivated by the works in real geometry related to Hilbert's $17^{\text{th}}$ problem; see Reznick's survey~\cite{reznick2000some} for more on this research area.
One difference between our focus here and much of the other literature on the SOS algorithm is that we are content with proving that
the algorithm supplies an \emph{approximation} to the true quantity, rather than exact convergence, but on the other hand are much more
stringent about using only very low degree (preferably constant or polylogarithmic in the number of variables).


\section{Sums of Squares Proofs and Algorithms}

One of the most common ways of proving that a quantity is non-negative is by expressing it as a \emph{Sum of Squares} (SOS).
For example, we can prove the Arithmetic-Mean Geometric-Mean inequality $ab \leq a^2/2 + b^2/2$ by the identity $a^2 + b^2 - 2ab = (a-b)^2$.
Thus a natural question, raised in the late $19^{\text{th}}$ century, was whether \emph{any} non-negative (possibly multivariate) polynomial can be written as a sum of squares of polynomials.
This was answered negatively by Hilbert in 1888, who went on to ask as his $17^{\text{th}}$ problem whether any such polynomial can be written as a sum of squares of \emph{rational}
functions.
A positive answer was given by Artin~\cite{Artin27}, and considerably strengthened by Krivine and Stengle.
In particular, the following theorem is a corollary of their results, which captures much of the general case.

\begin{theorem}[Corollary of the Positivstellensatz~\cite{Krivine64,Stengle74}]
  \label{thm:psatz}
  Let $P_1,\ldots,P_m\in \mathbb{R}[x]=\mathbb{R}[x_1,\ldots,x_n]$ be multivariate polynomials.
  Then, the system of polynomials equations $\mathcal{E} = \{P_1=0,\ldots,P_m=0\}$ has no solution over $\mathbb{R}^n$ if and only if, there exists polynomials $Q_1,\ldots,Q_m\in \mathbb{R}[x]$ such that $S\in\mathbb{R}[x]$ is a \emph{sum of squares} of polynomials and
  \begin{equation}
    -1 = S + \sum Q_i \cdot P_i  \;. \label{eq:Psatz}
  \end{equation}
\end{theorem}

We say that the polynomials $S,Q_1,\ldots,Q_m$ in the conclusion of the theorem form an \emph{SOS proof} refuting the system of polynomial equations\footnotemark\ $\mathcal{E}$.
\footnotetext{
In this survey we restrict attention to polynomial \emph{equalities} as opposed to \emph{inequalities},
which turns out to be without loss of generality for our purposes.
If we have a system of polynomial inequalities $\{P_1\ge 0,\ldots,P_m\ge 0\}$ for $P_i\in \mathbb{R}[x]$, the Positivstellensatz certificates of infeasibility take the form $-1=\sum_{\alpha \subseteq [n]} Q_\alpha P_\alpha$, where each $Q_\alpha\in\mathbb{R}[X]$ is a sum of squares and $P_\alpha=\prod_{i\in \alpha} P_i$.
However, we can transform inequalities $\{P_i\ge 0\}$ to equivalent equalities  $\{P_i'=P_i-y_i^2=0\}$, where $y_1,\ldots,y_m$ are fresh variables.
This transformation makes it only easier to find certificates, because $\sum_{\alpha\subseteq [n]} Q_\alpha P_\alpha = S' + \sum_i Q'_i P'_i$ for $S'=\sum_{\alpha\subseteq [n]} Q_\alpha y_\alpha^2$, where $y_\alpha=\prod_{i\in \alpha} y_i$.
It also follows that the transformation can only reduce the degree of SOS refutations.}
Clearly the existence of such polynomials implies that $\mathcal{E}$ is unsatisfiable---the interesting part of Theorem~\ref{thm:psatz} is the other direction.
We say that a SOS refutation $S_1,Q_1,...,Q_m$ has \emph{degree} $\ell$ if the maximum degree of the polynomials $Q_iP_i$ involved in the proof is at most $\ell$~\cite{GrigorievV01}.
By writing down the coefficients of these polynomials, we see that a degree-$\ell$ SOS proof can be written using $mn^{O(\ell)}$ numbers.\footnote{It  can be shown that the decomposition of $S$ into sums of squares will not require more than $n^{\ell}$ terms; also in all the
settings we consider, there are no issues of accuracy in representing real numbers, and so a degree $\ell$-proof can be written down
using $mn^{O(\ell)}$ bits.}

In the following lemma, we will prove a special case of Theorem~\ref{thm:psatz}, where the solution set of $\mathcal{E}$ is a subset of the hypercube $\{\pm 1\}^n$.
Here, the degree of SOS refutations is bounded by $2n$.
(This bound is not meaningful computationally because the size of degree-$\Omega(n)$ refutations is comparable to the number of points in $\{\pm 1\}^n$.)

\begin{lemma}
  \label{lem:hypercube-pss}
  Let $\mathcal{E}=\{P_0=0,x_1^2-1=0,\ldots,x_n^2-1=0\}$ for some $P_0\in\mathbb{R}[x]$.
  Then, either the system $\mathcal{E}$ is satisfiable or it has a degree-$2n$ SOS refutation.
\end{lemma}

\begin{proof}
Suppose the system is not satisfiable, which means that $P_0(x)\neq 0$ for all $x\in\{\pm 1\}^n$.
Since $\{\pm 1\}^n$ is a finite set, we may assume $P_0^2 \ge 1$ over $\{\pm 1\}^n$.
Now interpolate the real-valued function $\sqrt{P_0^2-1}$ on $\{\pm1\}^{{}n}$ as a multilinear (and hence degree at most $n$) polynomial in $R\in\mathbb{R}[x]$.
Then, $P_0^2-1-R^2$ is a polynomial of degree at most $2n$ that vanishes over $\{\pm 1\}^n$. 
(Since we can replace $x_i^2$ by $1$ in any monomial, we can assume without loss of generality that $P_0$ is multilinear and hence has degree at most $n$.)
This means that we can write $P_0^2-1-R^2$  in the form $\sum_{i=1}^n Q_i \cdot (x_i^2-1)$ for polynomials $Q_i$ with $Q_i\le \deg 2n-2$.
(This fact can be verified either directly or by using that $x_1^2-1,\ldots,x_n^2-1$ is a Gr\"obner basis for $\{\pm 1\}^n$.)
Putting things together, we see that $-1=R^2+ (-P_0)\cdot P_0+\sum_{i=1}^n Q_i \cdot (x_i^2-1)$, which is a SOS refutation for $\mathcal{E}$ of the form in Theorem~\ref{thm:psatz}.
\end{proof}

\subsection{From proofs to algorithms}
\label{sec:SOS-alg}
The Sum of Squares algorithm is based on the following theorem, which was discovered in different forms by several researchers:

\begin{theorem}[SOS Theorem~\cite{Shor87,Nesterov00,Parrilo00,Lasserre01}, informally stated]
  \label{thm:SOS}
  If there is a degree-$\ell$ SOS proof refuting $\mathcal{E}=\{P_1=0,\ldots,P_m=0\}$, then such a proof can be found
  in $mn^{O(\ell)}$ time.
\end{theorem}

\begin{proof}[Proof sketch]
We can view a degree-$\ell$ SOS refutation $-1=S+\sum_i Q_iP_i$ for $\mathcal{E}$ as a system of linear equations in $mn^{O(\ell)}$ variables corresponding to the coefficients of the unknown polynomials $S,Q_1,\ldots,Q_m$.
We only need to incorporate the non-linear constraint that $S$ is a sum of squares.
But it is not hard to see that a degree-$\ell$ polynomial $S$ is a sum of squares if and only if there exists a positive-semidefinite matrix $M$ such that $S=\sum_{\alpha,\alpha'} M_{\alpha,\alpha'} x^\alpha x^{\alpha'}$, where $\alpha$ and $\alpha'$ range over all monomials $x^\alpha$ and $x^{\alpha'}$ of degree at most $\ell/2$.
Thus, the task of finding a degree-$\ell$ SOS refutation reduces to the task of solving linear systems of equations with the  additional constraint that matrix formed by some of the variables is positive-semidefinite.
\emph{Semidefinite programming} solves precisely this task and is computationally efficient.\footnote{
In this survey we ignore issues of numerical accuracy which turn out to be easily handled in our setting.}
\end{proof}

\begin{remark}[\emph{What does ``efficient'' mean?}]
In the applications we are interested in, the number of variables $n$ corresponds to our ``input size''.
The equation systems $\mathcal{E}$ we consider can always be solved via a ``brute force'' algorithm running in $\exp(O(n))$ time,
and so degree-$\ell$ SOS proofs become interesting when  $\ell$ is much smaller than $n$.
Ideally we would want $\ell = O(1)$, though $\ell=\mathrm{polylog}(n)$ or even, say, $\ell = \sqrt{n}$, is still interesting.
\end{remark}

Theorem~\ref{thm:SOS} yields the following \emph{meta algorithm} that can be applied on any problem of the form
\begin{equation}
  \min_{x \in \mathbb{R}^n\colon  P_1(x)=\dots=P_m(x)=0} P_0(x) \label{eq:poly-opt}
\end{equation}
where $P_0,P_1,\ldots,P_m\in \mathbb{R}[x]$ are polynomials.
The algorithm is parameterized by a number $\ell$ called its \emph{degree} and operates as follows:

\begin{center}
\fbox{\begin{minipage}{\textwidth}
\begin{center}\textbf{The degree-$\ell$ Sum-of-Squares Algorithm}  \end{center}
\noindent\textbf{Input:} Polynomials $P_0,\ldots,P_m\in\mathbb{R}[x]$ \\
\noindent\textbf{Goal:} Estimate $\min P_0(x)$ over all $x\in\mathbb{R}^n$ such that $ P_1(x)=\ldots=P_m(x)=0$ \\
\noindent\textbf{Operation:}  Output the smallest  value $\varphi^{(\ell)}$
such that there does \emph{not} exist a \mbox{degree-$\ell$} SOS proof refuting the system,
\begin{displaymath}
\{ P_0= \varphi^{(\ell)},P_1=0,\ldots, P_m(x) =0 \}\,. \footnotemark
\end{displaymath}
\end{minipage}}
\end{center}

\footnotetext{As in other cases, we are ignoring here issues of numerical accuracy.
Also, we note that when actually executing this algorithm, we will not need to check all the
(uncountably many) values $\varphi^{(\ell)}\in \mathbb{R}$, but it suffices
to enumerate over a sufficiently fine discretization of the interval $[-M,+M]$ for some number $M$ depending on the polynomials
$P_0,\ldots,P_m$.
This step can be carried out in polynomial time in all the settings we consider.}

We call $\varphi^{(\ell)}$ the \emph{degree-$\ell$ SOS estimate} for (\ref{eq:poly-opt}),
and by Theorem~\ref{thm:SOS} it can be computed in $n^{O(\ell)}$ time.
For the actual minimum value $\varphi$ of (\ref{eq:poly-opt}), the corresponding system of equations $\{ P_0 = \varphi, P_1=0,\ldots,P_m=0 \}$ is satisfiable, and hence in particular cannot be refuted by an SOS proof.
Thus, $\varphi^{(\ell)} \leq \varphi$ for any $\ell$.
Since higher degree proofs are more powerful (in the sense that they can refute more equations), it holds that
\[
\varphi^{(2)} \leq \varphi^{(4)} \leq  \varphi^{(6)} \leq \cdots \leq \min_{x\in\mathbb{R}^n\colon P_1(x)=\dots=P_m(x)=0} P_0(x)  \;.
\]
(We can assume degrees of SOS proofs to be even.)
As we've seen in Lemma~\ref{lem:hypercube-pss}, for the typical domains we are interested in Computer Science, such as when the set of solutions of $\{P_1=0,\ldots,P_m=0\}$ is equal to $\{ \pm1 \}^n$, this sequence is finite in the sense that $\varphi^{(2n)} = \min_{x \in \{\pm 1\}^n} P_0(x)$.

The SOS algorithm uses semidefinite programming in a much more general way than many previous algorithms such as~\cite{lovasz1979shannon,GoemansW95}.
In fact, the UGC meta-algorithm is the same as the base case (i.e., $\ell=2$) of the SOS algorithm.



Recall that the UGC and SSEH imply that in many settings, one cannot improve on the approximation guarantees of the UGC meta-algorithm without using $\exp(n^{\Omega(1)})$ time.
Thus in particular, if those conjectures are true then in those settings, using the SOS meta algorithm with degree, say, $\ell=10$ (or even $\ell=\mathrm{polylog}(n)$ or $\ell = n^{o(1)}$) will not yield significantly better guarantees than $\ell=2$.

\begin{remark}[Comparison with local-search based algorithms]
Another approach to optimize over non-linear problems such as (\ref{eq:poly-opt}) is to use local-search
algorithms such as gradient descent that make local improvement steps, e.g., in the direction of the gradient, until a local optimum is reached.
One difference between such local search algorithms and the SOS algorithm is that the latter sometimes succeeds in optimizing highly non-convex problems that have exponential number of local optima.
As an illustration, consider the polynomial $P(x) = n^4\sum_{i=1}^n (x_i^2-x_i)^2 + (\sum_{i=1}^n x_i)^2$.
  
Its unique global minimum is the point $x=0$, but it is not hard to see that it has
an exponential number of local minima (for every $x\in \{0,1\}^n$, $P(x) < P(y)$ for every $y$ with
$\|y-x\| \in [1/n,2/n]$, and so there must be a local minima in the ball of radius $1/n$ around $x$).
Hence, gradient descent or other such algorithms are extremely likely to get stuck in one of these suboptimal local minima.
However, since $P$ is in fact a sum of squares with constant term $0$, the degree-$4$ SOS algorithm will output $P$'s correct global minimum value.

\end{remark}

\subsection{Pseudodistributions and pseudoexpectations}
\label{sec:pseudo-dist}

Suppose we want to show that the level-$\ell$ SOS meta-algorithm achieves a good approximation of the minimum value of $P_0$ over the set $\mathcal{Z}=\{x\in\mathbb{R}^n \mid P_1(x)=\dots=P_m(x)=0\}$ for a particular kind of polynomials $P_0,P_1,\ldots,P_m\in\mathbb{R}[x]$.
Since the estimate $\varphi^{(\ell)}$ always lower bounds this quantity, we are to show that
\begin{equation}
  \min_{\mathcal{Z}} P_0 \leq f(\varphi^{(\ell)}) \label{eq:meta-works}
\end{equation}
for some particular function $f$ (satisfying $f(\varphi)\geq \varphi$) which captures our approximation guarantee.
(E.g., a  factor $c$ approximation corresponds to the function $f(\varphi) = c\varphi$.)

If we expand out the definition of $\varphi^{(\ell)}$, we see that to prove Equation~(\ref{eq:meta-works}) we need to show that for every $\varphi$ if there does not exists a degree-$\ell$ proof that $P_0(x) \neq \varphi$ for all $x\in\mathcal{Z}$, then there exists an $x\in \mathcal{Z}$ such that $P_0(x) \leq f(\varphi)$.
So, to prove a result of this form, we need to find ways to use the \emph{non-existence} of a proof.
Here, \emph{duality} is useful.
\begin{quote}
\emph{Pseudodistributions are the dual object to SOS refutations, and hence the \emph{non-existence} of a refutation implies the \emph{existence} of a pseudodistribution.}
\end{quote}

We now elaborate on this, and explain both the definition and intuition behind pseudodistributions.
In Section~\ref{sec:sos-cheeger} we will give a concrete example, by showing how one can prove that degree-$2$ SOS proofs
capture Cheeger's Inequality using such an argument.
Results such as the analysis of the Goemans-Williamson \ensuremath{\text{\textsc{Max Cut}}}\ algorithm~\cite{GoemansW95}, and the proof of Grothendieck's Inequality~\cite{grothendieck1953resume} can be derived using similar methods.

\begin{definition}
Let $\mathbb{R}[x]_{\ell}$ denote the set of polynomials in $\mathbb{R}[x]$
of degree at most $\ell$.
  A \emph{degree-$\ell$ pseudoexpectation operator for $\mathbb{R}[x]$} is a linear operator $\mathcal{L}$ that maps
polynomials in $\mathbb{R}[x]_{\ell}$ into $\mathbb{R}$ and satisfies that $\mathcal{L}(1) = 1$ and $\mathcal{L}(P^2) \geq 0$ for every polynomial $P$ of degree at most $\ell/2$.

\end{definition}

The term pseudoexpectation stems from the fact that for every distribution $\mathcal{D}$ over $\mathbb{R}^n$, we can obtain such an operator by choosing $\mathcal{L}(P)=\bbE_{\mathcal{D}} P$ for all $P\in\mathbb{R}[x]$.
Moreover, the properties $\mathcal{L}(1)=1$ and $\mathcal{L}(P^2)\ge 0$ turn out to capture to a surprising extent the properties of distributions and their expectations that we tend to use in proofs.
Therefore, we will use a notation and terminology for such pseudoexpectation operators that parallels the notation we use for distributions.
In fact, all of our notation can be understood by making the thought experiment that there exists a distribution as above and expressing all quantities in terms of low-degree moments of that distribution (so that they also make sense if we only have a pseudoexpectation operator that doesn't necessarily correspond to a distribution).

In the following, we present the formal definition of our notation.
We denote pseudoexpectation operators as $\pE_\mathcal{D}$, where $\mathcal{D}$ acts as index to distinguish different operators.
If $\pE_\mathcal{D}$ is a degree-$\ell$ pseudoexpectation operator for $\mathbb{R}[x]$, we say that $\mathcal{D}$ is a \emph{degree-$\ell$ pseudodistribution} for the indeterminates $x$.
In order to emphasize or change indeterminates, we use the notation $\pE_{y\sim \mathcal{D}} P(y)$.
In case we have only one pseudodistribution $\mathcal{D}$ for indeterminates $x$, we denote it by $\{x\}$.
In that case, we also often drop the subscript for the pseudoexpectation and write $\pE P$ for $\pE_{\{x\}}P$.

We say that a degree-$\ell$ pseudodistribution $\{x\}$ satisfies a system of polynomial equations $\{P_1=0,\ldots,P_m=0\}$ if $\pE Q\cdot P_i=0$ for all $i\in[m]$ and all polynomials $Q\in\mathbb{R}[x]$ with $\deg Q \cdot P_i\le \ell$.
We also say that $\{ x \}$ satisfies the constraint $\{ P(x) \ge  0 \}$ if there exists some sum-of-squares polynomial $S\in\mathbb{R}[x]$ such that $\{ x \}$ satisfies the polynomial equation $\{ P=S \}$.
It is not hard to see that if $\{ x \}$ was an actual distribution, then these definitions imply that all points in the support of the distribution satisfy the constraints.
We write $P\succcurlyeq 0$ to denote that $P$ is a sum of squares of polynomials, and similarly we write $P \succcurlyeq Q$ to denote $P-Q \succcurlyeq 0$.

The duality  between SOS proofs and pseudoexpectations is expressed in the following theorem.
We say that a system $\mathcal{E}$ of polynomial equations is \emph{explicitly bounded} if there exists a linear combination of the constraints in $\mathcal{E}$ that has the form $\{\sum_i x_i^2 + S= M\}$ for $M\in\mathbb{R}$ and $S\in\mathbb{R}[x]$ a sum-of-squares polynomial.
(Note that in this case, every solution $x\in\mathbb{R}^n$ of the system $\mathcal{E}$ satisfies $\sum_i x_i^2\le M$.)

\begin{theorem}
  \label{thm:duality}
  Let $\mathcal{E}=\{P_1=0,\ldots,P_m=0\}$ be a set of polynomial equations with $P_i\in\mathbb{R}[x]$.
  Assume that $\mathcal{E}$ is explicitly bounded in the sense above.
  Then, exactly one of the following two statements holds: (a) there exists a degree-$\ell$ SOS proof refuting~$\mathcal{E}$, or (b) there exists a degree-$\ell$ pseudodistribution $\{x\}$ that satisfies~$\mathcal{E}$.
\end{theorem}

\begin{proof}

First, suppose there exists a degree-$\ell$ refutation of the system $\mathcal{E}$, i.e., there exists polynomials $Q_1,\ldots,Q_m\in\mathbb{R}[x]$ and a sum-of-squares polynomial $R\in\mathbb{R}[x]$ so that $-1=R+\sum_iQ_iP_i$ and $\deg Q_iP_i\le \ell$.
Let $\{x\}$ be any pseudodistribution.
We are to show that $\{x\}$ does not satisfy $\mathcal{E}$.
Indeed, $\pE \sum_i Q_i P_i = -\pE 1 -\pE R\le -1$, which means that $\pE Q_i P_i\neq 0$ for at least one $i\in[m]$.
Therefore, $\{x\}$ does not satisfy $\mathcal{E}$.

Next, suppose there does not exist a degree-$\ell$ refutation of the system $\mathcal{E}$.
We are to show that there exists a pseudodistribution that satisfies $\mathcal{E}$.
Let $\mathcal{C}$ be the cone of all polynomials of the form $R+\sum_i Q_iP_i$ for sum-of-squares $R$ and polynomials $Q_i$ with $\deg Q_i P_i\le \ell$.
Since $\mathcal{E}$ does not have a degree-$\ell$ refutation, the constant polynomial $-1$ is not contained in $\mathcal{C}$.
We claim that from our assumption that the system~$\mathcal{E}$ is explicitly bounded it follows that $-1$ also cannot lie on the boundary of $\mathcal{C}$.
Assuming this claim, the hyperplane separation theorem implies that there exists a linear form $L$ such that $L(-1)<0$ but $L(P)\ge 0$ for all $P\in\mathcal{C}$.
By rescaling, we may assume that $L(1)=1$.
Now this linear form satisfies all conditions of a pseudoexpectation operator for the system~$\mathcal{E}$.

\emph{Proof of claim.}
We will show that if $-1$ lies on the boundary of $\mathcal{C}$, then also $-1\in\mathcal{C}$.
If $-1$ is on the boundary of $\mathcal{C}$, then there exists a polynomial~$P\in\mathbb{R}[X]_\ell$ such that $-1+\varepsilon P\in \mathcal{C}$ for all $\varepsilon>0$ (using the convexity of $\mathcal{C}$).
Since $\mathcal{E}$ is explicitly bounded, for every polynomial~$P\in\mathbb{R}[X]_\ell$, the cone $\mathcal{C}$ contains a polynomial of form $N-P-R$ for a sum-of-square $R$ and a number $N$.
(Here, the polynomial~$N-P-R\in\mathcal{C}$ is a certificate that $P\le N$ over the solution set of $\mathcal{E}$.
Such a certificate is easy to obtain when $\mathcal{E}$ is explicitly bounded.
We are omitting the details.)
At this point, we see that $-1$ is a nonnegative combination of the polynomials~$-1+\varepsilon P$, $N-P-R$, and $R$ for $\varepsilon<1/N$.
Since these polynomials are contained in $\mathcal{C}$, their nonnegative combination $-1$ is also contained in the cone~$\mathcal{C}$.
\end{proof}

\paragraph{Recipe for using pseudoexpectations algorithmically.}
In many applications we will use the following dual form of the SOS algorithm:

\begin{center}
\fbox{\begin{minipage}{\textwidth}
\begin{center}\textbf{The degree-$\ell$ Sum-of-Squares Algorithm (dual form)}  \end{center}
\noindent\textbf{Input:} Polynomials $P_0,\ldots,P_m\in\mathbb{R}[x]$ \\
\noindent\textbf{Goal:} Estimate $\min P_0(x)$ over all $x$ with $P_1(x)=\ldots=P_m(x)=0$ \\
\noindent\textbf{Operation:}  Output the smallest value $\varphi^{(\ell)}$
such that there is a degree-$\ell$ pseudodistribution $\{ x \}$ satisfying the system,
\begin{displaymath}
\{ P_0= \varphi^{(\ell)},P_1=0,\ldots, P_m(x) =0 \}\,.
\end{displaymath}

\end{minipage}}
\end{center}


Theorem~\ref{thm:duality} shows that in the cases we are interested in, both variants of the SOS algorithm
will output the same answer.
Regardless, a similar proof to that of Theorem~\ref{thm:SOS} shows that the dual form of the SOS algorithm
can also be computed in time $n^{O(\ell)}$.
Thus, when using the SOS meta-algorithm, instead of trying to argue from the non-existence of a proof,
we will use the existence of a pseudodistribution.
Specifically, to show that the algorithm provides an $f(\cdot)$ approximation in the sense of (\ref{eq:meta-works}),
what we need to show is that given a degree-$\ell$ pseudodistribution
$\{ x \}$ satisfying the system $\{ P = \varphi, P_1= 0,\ldots,P_m=0\}$, we can find some particular $x^*$ that satisfies
$P(x^*) \leq f(\varphi)$.
Our approach to doing so (based on the authors' paper with Kelner~\cite{BarakKS14}) can be summarized as follows:

\begin{quote}
\emph{Solve the problem pretending that $\{ x \}$ is an \emph{actual distribution} over solutions, and if all the steps you used have low-degree SOS proofs, the solution still works even when $\{x\}$ is a low-degree \emph{pseudodistribution}.}
\end{quote}

It may seem that coming up with an algorithm for the actual distribution case is trivial, as any element in the support of the distribution would be a good solution.
However note that even in the case of a real distribution, the algorithm does not get sampling access to the distribution, but only access to its low-degree moments.
Depending on the reader's temperament, the above description of the algorithm, which ``pretends''
pseudodistributions  are real ones, may sound tautological or just wrong.
Hopefully it will be clearer after the next two sections, where we use this approach to show how the SOS algorithm can match the guarantee of Cheeger's Inequality for computing the expansion, to find planted sparse vectors in random subspaces, and to approximately recover sparsely used dictionaries.


\section{Approximating expansion via sums of squares}
\label{sec:sos-cheeger}

Recall that the \emph{expansion}, $\phi_G$, of a $d$-regular graph $G=(V,E)$ is the minimum of $\phi_G(S)=|E(S,V\setminus S)||/(d|S|)$ over all sets $S$ of size at most $|V|/2$.
Letting $x=\mathds 1_{S}$ be the characteristic vector\footnotemark{} of the set $S$ the expression $|E(S,V\setminus S)|$ can be written as $\sum_{\{i,j\}\in E} (x_i-x_j)^2$ which is a quadratic polynomial in $x$.
\footnotetext{The $i$-th coordinate of vector $\mathds 1_S$ is equal $1$ if $i\in S$ and equal $0$ otherwise.}
Therefore, for every $k$, computing the value $\phi_G(k) = \min_{|S|=k}|E(S,V\setminus S)|/(dk)$ can be phrased as the question of minimizing a polynomial $P_0$ over the set of $x$'s satisfying the equations $\{ x_i^2 - x_i =0\}_{i=1}^n$ and $\{\sum_{i=1}^n x_i = k\}$.
Let  $\cramped{\phi_G^{\scriptscriptstyle (\ell(k))}}$ be the degree-$\ell$ SOS estimate for $\phi_G(k)$.
We call $\cramped{\phi_G^{\scriptscriptstyle (\ell)}} = \min_{k \leq n/2} \phi_G(k)$ the degree-$\ell$ SOS estimate for $\phi_G$.
Note that $\cramped{\phi_G^{\scriptscriptstyle (\ell)}}$ can be computed in $\cramped{n^{O(\ell)}}$ time.
For the case $\ell=2$, the following theorem describes the approximation guarantee of the estimate $\cramped{\phi_G^{\scriptscriptstyle (\ell)}}$.

\begin{theorem} \label{thm:sparsest-cut}
There exists an absolute constant $c$ such that for every graph $G$
\begin{equation}
\phi_G \leq c\sqrt{\phi_G^{(2)}} \label{eq:sparsest-cut-sdp}
\end{equation}
\end{theorem}

Before we prove Theorem~\ref{thm:sparsest-cut}, let us discuss its significance.
Theorem~\ref{thm:sparsest-cut} is essentially a restatement of Cheeger's Inequality in the SOS language---the degree $2$-SOS algorithm is the UGC meta algorithm which is essentially the same as the algorithm based on the second-largest eigenvalue.\footnote{
The second-largest eigenvalue is directly related to the minimum value of $\varphi$ such that there exists a degree-$2$ pseudodistribution satisfying the more relaxed system $\{\sum_{\{ij\}\in E} (x_i-x_j)^2=\varphi \cdot  d n/2, \sum_i x_i=n/2, \sum_i x_i^2 =n/2\}$.}
There are examples showing that~(\ref{eq:sparsest-cut-sdp}) is tight, and so we cannot get better approximation
using degree $2$ proofs.
But can we get a better estimate using degree $4$ proofs?
Or degree $\log n$ proofs?
We don't know the answer, but if the Small-Set Expansion Hypothesis is true, then beating the estimate (\ref{eq:sparsest-cut-sdp}) is \ensuremath{\mathbf{NP}} -hard,  which means (under standard assumptions) that to do so we will need to use
proofs of degree at least $n^{\Omega(1)}$.

This phenomenon repeats itself in other problems as well.
For example, for both the Grothendieck Inequality and the \ensuremath{\text{\textsc{Max Cut}}}\  problems, the SSEH (via the UGC) predicts that beating the estimate obtained by
degree-$2$ proofs will require degree $\ell = n^{\Omega(1)}$.
As in the case of expansion, we have not been able to confirm or refute these predictions.
However, we will see some examples where using higher degree proofs \emph{does} help, some of them suspiciously close
in nature to the expansion problem.

One such example comes from the beautiful work of Arora, Rao and Vazirani~\cite{AroraRV09} who showed that
\[
\phi_G \leq O(\sqrt{\log n}) \cdot \phi_G^{(6)} \;,
\]
which is better than the guarantee of Theorem~\ref{thm:sparsest-cut} for $\phi_G \ll 1/\log n$.
However, this is not known to contradict the SSEH or UGC, which apply to the case when $\phi_G$ is a small constant.

As we will see in Section~\ref{sec:SSEvsSOS}, for the small set expansion problem of approximating $\phi_G(S)$ for small sets $S$, we can beat the degree $2$ bounds with degree $\ell = n^{\tau}$ proofs where $\tau$ is a parameter tending to zero with the parameter $\varepsilon$ of the SSEH~\cite{AroraBS10}.
This yields a sub-exponential algorithm for the small-set expansion problem (which can be extended to the \textsc{Unique Games} problem as well) that ``barely misses'' refuting the SSEH and UGC.
We will also see that degree $O(1)$ proofs have surprising power in other settings that are closely related to the SSEH/UGC, but again at the moment still fall short of refuting those conjectures.

\subsection{Proof of Theorem~\ref{thm:sparsest-cut}}
\label{sec:sparsest-cut:proof}

This proof is largely a reformulation of the standard proof of a discrete variant of Cheeger's Inequality,
phrased in the SOS language of pseudodistributions, and hence is included here mainly to help clarify these notions,
and to introduce a tool--- sampling from a distribution matching first two moments of a pseudodistribution--- that will be useful for us later on.
By the dual formulation, to prove Theorem~\ref{thm:sparsest-cut} we need to show that given a
pseudodistribution $\{ x \}$ over characteristic vectors of size-$k$ sets $S$ of size $k \leq n/2$ with
$|E(S,V\setminus S)|= \varphi dk$, we can find a particular set $S^*$ of size
at most $n/2$
such that $E(S^*,V\setminus S^*) \leq O(\sqrt{\varphi})d|S^*|$.
For simplicity, we consider the case $k=n/2$  (the other cases can be proven in a very similar way).
The distribution $\{ x \}$ satisfies the constraints
$\{ \sum x_i = n/2 \}$, $\{ x_i^2 = x_i \}$ for all $i$,
and $\{ \sum_{\{i,j\}\in E} (x_i-x_j)^2 = \varphi d \sum_i x_i \}$.
The algorithm to find $S^*$ is quite simple:

\begin{enumerate}
\item Choose $(y_1,\ldots,y_n)$ from a random Gaussian distribution with the same quadratic moments as $\{x\}$ so that $\bbE y_i = \pE x_i$
and $\bbE y_iy_j = \pE x_ix_j$ for all $i,j\in[n]$.
(See details below.)
\item Output the set $S^* = \{ i \mid y_i \geq 1/2 \}$ (which corresponds to the 0/1 vector closest to $y$).
\end{enumerate}
We remark that the set produces by the algorithm might have cardinality larger than $n/2$, in which case we will take the complement of $S^*$.

\paragraph{Sampling from a distribution matching two moments.}
We will first give a constructive proof the well-known fact that for every distribution over $\mathbb{R}^n$, there exists an $n$-dimensional Gaussian distribution with the same quadratic moments.
Given the moments of a distribution $\{ x \}$ over $\mathbb{R}^n$, we can sample a Gaussian distribution $\{ y \}$ matching the first two moments of $\{x \}$ as follows.
First, we can assume $\bbE x_i = 0$ for all $i$ by shifting variables if necessary.
Next, let $v^1,\ldots,v^n$ and $\lambda_1,\ldots,\lambda_n$ be the eigenvectors and eigenvalues
of the matrix $M_{i,j} = \bbE x_i x_j$.
(Note that $M$ is positive semidefinite and so $\lambda_1,\ldots,\lambda_n \geq 0$.)
Choose i.i.d random standard Gaussian variables $w_1,\ldots,w_n$ and define $y = \sum_k \sqrt{\lambda_k} w_k v^k$.
Since $\bbE w_kw_{k'}$ equals $1$ if $k=k'$ and equals $0$ otherwise,
\[
\bbE y_iy_j = \sum_k \lambda_k (v^k)_i(v^k)_j = M_{i,j} \;.
\]
One can verify that if $\{ x \}$ is a degree-$2$ pseudodistribution then the second moment matrix $M$ of the shifted
version of $x$ (such that $\pE x_i = 0$ for all $i$) is positive-semidefinite, and hence the above can
be carried for pseudodistributions of degree at least $2$ as well.
Concretely, if we let $\bar x = \pE x$ be the mean of the pseudodistribution, then $M=\pE (x-\bar x)(x-\bar x)^\top$.
This matrix is positive semidefinite because every test vector $z\in\mathbb{R}^n$ satisfies $z^\top M z = \pE \bigl( z^\top (x-\bar x)\bigr)^2\ge 0$.

\paragraph{Analyzing the algorithm.}
The analysis is based on the following two claims:
(i) the set $S^*$ satisfies $n/3\le \lvert  S^* \rvert\le 2n/3 $ with constant probability
and (ii) in expectation $|E(S^*,V\setminus S^*)| \leq O(\sqrt{\varphi} dn)$.

We will focus on two extreme cases that capture the heart of the arguments for the claims.
In the first case, all variables $y_i$ have very small variance so that $\bbE y_i^2\approx (\bbE y_i)^2$.
In this case, because our constraints imply that $\bbE y_i^2 = \bbE y_i$, every variable satisfies either $\bbE y_i^2\approx 0$ or $\bbE y_i^2\approx 1$, which means that the distribution of the set $S^*$ produced by the algorithm is concentrated around a particular set, and it is easy to verify that this set satisfies the two claims.
In the second, more interesting case, all variables $y_i$ have large variance, which means $\bbE y_i^2=1/2$ in our setting.

In this case, each event $\{y_i\ge 1/2\}$ has probability $1/2$ and therefore $\bbE\lvert  S^\ast \rvert=n/2$.
Using that the quadratic moments of $\{y\}$ satisfy $\bbE\sum_i y_i=n/2$ and $\bbE(\sum_i y_i)^2=(n/2)^2$, one can show that these events cannot be completely correlated, which allows us to control the probability of the event $n/3\le \lvert  S^* \rvert\le 2n/3$ and establishes (i).
For the second claim, it turns out that by convexity considerations it suffices to analyze the case that all edges contribute equally to the term $\tfrac{1}{\lvert  E \rvert}\sum_{\{i,j\}\in E} \pE (x_i-x_j)^2 = \varphi\,,$
so that $\pE (x_i-x_j)^2 = \varphi$ for all $\{ i,j \} \in E$.
So we see that $\{y_i,y_j\}$ is a 2-dimensional Gaussian distribution with mean $(\tfrac12,\tfrac12)$ and covariance
\begin{math}
  \tfrac14
  \left (\begin{smallmatrix}
    1 & 1-2\varphi \\
    1-2\varphi & 1
  \end{smallmatrix}
  \right )
\end{math}
Thus, in order to bound the expected value of
\begin{math}
  |E(S^*,V \setminus S^*)|\,, 
\end{math}
we need to bound the probability of the event ``$y_i \geq 1/2$ and $y_j < 1/2$'' for this particular Gaussian distribution, which amounts to a not-too-difficult calculation that indeed yields an upper bound of $O(\sqrt{\varphi})$ on this probability. \qed


\section{Machine learning with Sum of Squares}
\label{sec:ML}

In this section, we illustrate the computational power of the sum-of-squares method with applications to two basic problems in unsupervised learning.
In these problems, we are given samples of an unknown distribution from a fixed, parametrized family of distributions and the goal is to recover the unknown parameters from these samples.
Despite the average-case nature of these problems, most of the analysis in these applications will be for deterministic problems about polynomials that are interesting in their own right.

\medskip

The first problem is \textsc{sparse vector recovery}.
Here, we are given a random basis of a $d$-dimensional linear subspace $U\subseteq \mathbb{R}^n$ of the form
\begin{displaymath}
  U=\Span\{ \cramped{x^{\scriptscriptstyle (0)}},\cramped{x^{\scriptscriptstyle (1)}},\ldots,\cramped{x^{\scriptscriptstyle (d)}} \}\,,
\end{displaymath}
where $\cramped{x^{\scriptscriptstyle (0)}}$ is a sparse vector and $\cramped{x^{\scriptscriptstyle (1)}},\ldots,\cramped{x^{\scriptscriptstyle (d)}}$ are independent standard Gaussian vectors.
The goal is to reconstruct the vector $\cramped{x^{\scriptscriptstyle (0)}}$.
This is a natural problem in its own right, and is also a useful subroutine in various settings; see~\cite{DemanetH13}.
Demanet and Hand~\cite{DemanetH13}  gave an algorithm (based on~\cite{SpielmanWW12}) that recovers $\cramped{x^{\scriptscriptstyle (0)}}$ by searching for the vector $x$ in $U$ that maximizes $\|x\|_{\infty}/\|x\|_1$ (which can be done efficiently by $n$ linear programs).
It is not hard to show that $\cramped{x^{\scriptscriptstyle (0)}}$ has to have  less than $|n|/\sqrt{d}$ coordinates for it to be maximize this ratio,\footnote{See Lemma~\ref{lem:dimbound} below for a related statement.} and hence this was a limitation
of prior techniques.
In contrast, as long as $d$ is not too large (namely, $d = O(\sqrt{n})$), the SOS method can recover $\cramped{x^{\scriptscriptstyle (0)}}$ as long as it has less than $\varepsilon n$ coordinates for some constant $\varepsilon > 0$~\cite{BarakKS14}.

The second problem we consider is \textsc{sparse dictionary learning}, also known as \textsc{sparse coding}.
Here, we are given independent samples $\cramped{y^{\scriptscriptstyle (1)}} ,\ldots, \cramped{y^{\scriptscriptstyle (R)}}\in\mathbb{R}^n$ from an unknown distribution of the form $\{y = A x \}$, where $A\in \mathbb{R}^{n\times m}$ is a matrix and $x$ is a random $m$-dimensional vector from a distribution over sparse vectors.
This problem, initiated by the work Olshausen and Field~\cite{olshausen1996emergence} in computational neuroscience, has found a variety of uses in machine learning, computer vision, and image processing (see, e.g.~\cite{AgarwalANT13} and the references therein).
The appeal of this problem is that intuitively data should be sparse in the ``right'' representation
(where every coordinate corresponds to a meaningful feature),
and finding this representation can be a useful first step for further processing, just as representing sound or image data in the Fourier or Wavelet bases is often a very useful primitive.
While there are many heuristics use to solve this problem, prior works giving rigorous recovery guarantees such as~\cite{SpielmanWW12,AgarwalANT13,AroraGM13} all required the vector $x$ to be \emph{very} sparse, namely less than $\sqrt{n}$ nonzero entries.\footnote{If the distribution $x$ consists of $m$ independent random variables then better guarantees can be achieved using \emph{Independent Component Analysis (ICA)}~\cite{comon1994independent}. See~\cite{GoyalVX13} for the current state of art in this setting.
However we are interested here in the more general case.}
In contrast, the SOS method can be used to approximately recover the dictionary matrix $A$ as long as $x$ has $o(n)$ nonzero (or more generally, significant) entries~\cite{BarakKS15}.

\subsection{Sparse vector recovery}
\label{sec:sparse-vector-recovery}

We say a vector $x$ is $\mu$-sparse if the 0/1 indicator $\mathds 1_{\supp x}$ of the support of $x$ has norm-squared $\mu=\lVert  \mathds 1_{\supp x} \rVert^2_2$.
The ratio $\mu/\lVert  \mathds 1 \rVert_2^2$ is the fraction of non-zero coordinates in $x$.

\begin{theorem}
  \label{thm:sparse-recovery}
  There exists a polynomial-time approximation algorithm for \textsc{sparse vector recovery} with the following guarantees:
  Suppose the input of the algorithm is an arbitrary basis of a $d+1$-dimensional linear subspace $U\subseteq \mathbb{R}^n$ of the form $U=\Span\{\cramped{x^{\scriptscriptstyle (0)}},\cramped{x^{\scriptscriptstyle (1)}}\ldots,\cramped{x^{\scriptscriptstyle (d)}}\}$ such that $\cramped{x^{\scriptscriptstyle (0)}}$ is a $\mu$-sparse unit vector with $\mu \leq \varepsilon\cdot \lVert  \mathds 1 \rVert_2^2$ and $\cramped{x^{\scriptscriptstyle (1)}},\ldots,\cramped{x^{\scriptscriptstyle (d)}}$ are standard Gaussian vectors orthogonal to $\cramped{x^{\scriptscriptstyle (0)}}$ with $d\ll \sqrt n$.
  Then, with probability close to $1$, the algorithm outputs a unit vector~$x$ that has correlation $\langle  x,\cramped{x^{\scriptscriptstyle (0)}} \rangle^2\ge 1- O(\varepsilon)$ with $\cramped{x^{\scriptscriptstyle (0)}}$.
\end{theorem}

Our algorithm will follow the general recipe we described in Section~\ref{sec:pseudo-dist}:

\begin{quote}
\emph{Find a system of polynomial equations $\mathcal{E}$ that captures the intended solution $\cramped{x^{\scriptscriptstyle (0)}}$, then pretend you are given a distribution $\{ u \}$ over solutions of $\mathcal{E}$ and show how you could recover a single solution $u^*$ from the low order moments of $\{ u \}$.}
\end{quote}

Specifically,  we come up with a system $\mathcal{E}$ so that desired vector $\cramped{x^{\scriptscriptstyle (0)}}$ satisfies all equations, and it is essentially the only solution to $\mathcal{E}$.
Then, using the SOS algorithm, we compute a degree-$4$ pseudodistribution $\{ u \}$ that satisfies $\mathcal{E}$.
Finally, as in Section~\ref{sec:sparsest-cut:proof}, we sample a vector $u^*$ from a Gaussian distribution that has the same quadratic moments as the pseudodistribution $\{ u \}$.

\paragraph{How to encode this problem as a system of polynomial equations?}

By Cauchy--Schwarz, any $\mu$-sparse vector $x$ satisfies $\lVert  x \rVert_2^2 \le \lVert  x \rVert^2_{2p} \cdot  \lVert  \mathds 1_{\supp x} \rVert_q=\lVert  x \rVert^2_{2p} \cdot \mu^{1-1/p}$ for all $p,q\ge 1$ with $1/p+1/q=1$.
In particular, for $p=2$, such vectors satisfy $\lVert x  \rVert_4^4\ge \lVert  x \rVert_2^4/\mu$.
This fact motivates our encoding of \textsc{sparse vector recovery} as a system of polynomial equations.
If the input specifies subspace $U\subseteq \mathbb{R}^n$, then we compute the projector $P$ into the subspace $U$ and choose the following polynomial equations:
$\lVert  u \rVert_2^2=1$ and $\lVert  P u \rVert_4^4=1/\mu_0$, where $\mu_0=\lVert  {}\cramped{x^{\scriptscriptstyle (0)}}  \rVert_2^4/\lVert  {}\cramped{x^{\scriptscriptstyle (0)}} \rVert_4^4$.
(We assume here the algorithm is given $\mu_0\le \mu$ as input, as we can always guess a sufficiently close approximation to it.)
%

\paragraph{Why does the sum-of-squares method work?}

The analysis of algorithm has two ingredients.
The first ingredient is a structural property about projectors of random subspaces.

\begin{lemma}
    \label{lem:random-norm-bound}
  Let $U'\subseteq \mathbb{R}^n$ be a random $d$-dimensional subspace with $d\ll \sqrt n$ and let $P'$ be the projector into $U'$.
  Then, with high probability, the following sum-of-squares relation over $\mathbb{R}[u]$ holds for $\mu' \ge \Omega(1)\cdot \lVert  \mathds 1 \rVert_2^2$,
  \begin{displaymath}
    \lVert  P'u \rVert_4^4\preccurlyeq \lVert  u \rVert_2^4/\mu'
    \,.
  \end{displaymath}
\end{lemma}

\begin{proof}[Proof outline]
We can write $P' = B^\top B$ where $B$ is a $d\times n$ matrix whose rows are an orthogonal basis for the subspace $U'$.  Therefore, $P'u = B^\top x$ where $x = Bu$, and so to prove Lemma~\ref{lem:random-norm-bound}
it suffices to show that under these conditions, $\|B^{\top}x\|_4^4 \preccurlyeq O(\|x\|_2^4 / \|\mathds 1\|_2^4)$.
The matrix $B^{\top}$ will be very close to having random independent Gaussian entries,
and hence, up to scaling, $\|B^\top x\|_4^4$ will be (up to scaling), close to
$Q(x) = \tfrac{1}{n}\sum \langle w_i,x \rangle^4$ where $w_1,\ldots,w_d \in \mathbb{R}^d$ are chosen independently at random from
the standard Gaussian distribution.
The expectation of $\langle w,x \rangle^4$ is equal  $3\sum_{i,j} x_i^2x_j^2 = 3\|x\|_2^4$.   
Therefore, to prove the lemma, we need to show that for $n\gg d^2$, the polynomial $Q(x)$ is with high probability close to its expectation, in the sense that the $d^2\times d^2$ matrix corresponding to $Q$'s coefficients is close to its expectation in the spectral norm.
This follows from standard matrix concentration inequalities, see~\cite[Theorem~7.1\footnotemark]{BarakBHKSZ12}).
\end{proof}
\footnotetext{The reference is for the arxiv version \texttt{arXiv:1205.4484v2} of the paper.}

The following lemma is the second ingredient of the analysis of the algorithm.

\begin{lemma}
\label{lem:sparse-vector-recovery}
  Let $U'\subseteq \mathbb{R}^n$ be a linear subspace and let $P'$ be the projector into $U'$.
  Let $\cramped{x^{\scriptscriptstyle (0)}}\in\mathbb{R}^n$ be a $\mu$-sparse unit vector orthogonal to $U'$ and let $U=\Span\{\cramped{x^{\scriptscriptstyle (0)}}\}\oplus U'$ and $P$ the projector on $U$.
  Let $\{u\}$ be a degree-$4$ pseudodistribution that satisfies the constraints $\{ \lVert  u \rVert_2^2=1 \}$ and $\{ \lVert  Pu \rVert_4^4=1/\mu_0 \}$, where $\mu_0=\lVert \cramped{x^{\scriptscriptstyle (0)}} \rVert_2^4/\lVert  \cramped{x^{\scriptscriptstyle (0)}} \rVert_4^4\le \mu$.
  Suppose $\lVert  P' u \rVert_4^4\preccurlyeq \lVert  u \rVert_2^4/\mu'$ is a sum-of-squares relation in $\mathbb{R}[u]$.
  Then, $\{u\}$ satisfies
  \begin{displaymath}
    \pE \lVert  P' u \rVert_2^2\le 4\bigl(\tfrac{\mu}{\mu'}\bigr)^{1/4}\,.
  \end{displaymath}
\end{lemma}

Note that the conclusion of Lemma~\ref{lem:sparse-vector-recovery} implies that a vector $u^*$ sampled from a Gaussian distribution with the same quadratic moments as the computed pseudodistribution also satisfies $\bbE_{u^*} \lVert  P' u^* \rVert_2^2\le 4(\mu/\mu')^{1/4}$ and $\bbE\lVert  u^\ast \rVert_2^2=1$.
By Markov inequality, $\lVert  u^*-\cramped{x^{\scriptscriptstyle (0)}}  \rVert_2^2\le 16(\mu/\mu')^{1/4}$ holds with probability at least $3/4$.
Since $u^\ast$ is Gaussian, it satisfies $\lVert u^\ast \rVert_2^2\ge 1/4$ with probability at least $1/2$.
If both events occur, which happens with probability at least $1/4$, then $\langle  u^\ast,\cramped{x^{\scriptscriptstyle (0)}} \rangle^2\ge (1-O(\mu/\mu'))\lVert  u^\ast \rVert_2^2$, thus establishing Theorem~\ref{thm:sparse-recovery}.

\paragraph{Proof of Lemma~\ref{lem:sparse-vector-recovery}}

There are many ways in which pseudodistributions behave like actual distributions, as far as low degree polynomials are concerned. To prove Lemma~\ref{lem:sparse-vector-recovery},
we need to establish the following two such results:

\begin{lemma}[H\"older's inequality for pseudoexpectation norms]
  Suppose $a$ and $b$ are nonnegative integers that sum to a power of $2$.
  Then, every degree-$(a+b)$ pseudodistribution $\{u,v\}$ satisfies
  \begin{displaymath}
    \pE \bbE_i u_i^a v_i^b \le \left(\pE \bbE_i u_i^{a+b}\right)^{a/(a+b)} \cdot \left(\pE \bbE_i v_i^{a+b}\right)^{b/(a+b)} \,.
  \end{displaymath}
\end{lemma}

\begin{proof}[Proof sketch]
The proof of the general case follows from the case $a=b=2$ by an inductive argument.
The proof for the case $a=b=1$ follows from the fact that the polynomial $\alpha \bbE_iu_i^2+\beta\bbE_iv_i^2-\sqrt{\alpha\beta}\bbE_iu_iv_i\in\mathbb{R}[u,v]$ is a sum of squares for all $\alpha,\beta\ge 0$ and choosing $\alpha=1/\pE \bbE_i u_i^2$ and $\beta=1/\pE\bbE_i v_i^2$
\end{proof}

\begin{lemma}[Triangle inequality for pseudodistribution $\ell_4$ norm]
  \label{lem:pseudo-hoelder}
  Let $\{u,v\}$ be a degree-$4$ pseudodistribution.
  Then,
  \begin{displaymath}
    \left(\pE \lVert  u+v \rVert_4^4\right)^{1/4} \le \left(\pE \lVert  u \rVert_4^4\right)^{1/4} + \left(\pE \lVert  v \rVert_4^4\right)^{1/4} \,.
  \end{displaymath}
\end{lemma}

\begin{proof}
  The inequality is invariant with respect to the measure used for the inner norm $\lVert  \cdot \rVert_4$.
  For simplicity, suppose $\lVert  x \rVert_4^4=\bbE x_i^4$.
  Then, $\lVert  u+v \rVert_4^4=\bbE_iu_i^4 + 4 \bbE_i u_i^3v_i + 6\bbE_iu_iv_i^3+\bbE_iv_i^4$.
  Let $A=\pE \bbE_iu_i^4$ and $B=\pE \bbE_iv_i^4$.
  Then, Lemma~\ref{lem:pseudo-hoelder} allows us to bound the pseudoexpectations of the terms $\bbE_iu_i^av_i^b$, so that as desired
  \begin{displaymath}
    \pE \lVert  u+v \rVert_4^4\le A + 4 A^{3/4} B^{1/4} + 6 A^{1/2}B^{1/2} + 4 A^{1/3}B^{3/4} + B = (A^{1/4}+B^{1/4})^4\,.\qedhere
  \end{displaymath}
\end{proof}

We can now prove Lemma~\ref{thm:sparse-recovery}.
Let $\alpha_0=\langle  u,{}\cramped{x^{\scriptscriptstyle (0)}} \rangle\in\mathbb{R}[u]$.
By construction, the polynomial identity $\lVert Pu\rVert_4^4 = \lVert \alpha_0 {}\cramped{x^{\scriptscriptstyle (0)}} + P' u\rVert_4^4$ holds over $\mathbb{R}[u]$.
By the triangle inequality for pseudodistribution $\ell_4$ norm,
for $A=\pE \alpha_0^4 \lVert  {}\cramped{x^{\scriptscriptstyle (0)}} \rVert_4^4$ and $B=\pE \lVert  P'u \rVert_4^4$
\begin{displaymath}
\cramped{\bigl( \tfrac1{\mu_0} \bigr)^{1/4}}=\bigl(\pE \lVert  P u \rVert_4^4\bigr)^{1/4} \le A^{1/4} + B^{1/4}
\end{displaymath}
By the premises of the lemma, $A= \pE \alpha_0^4/\mu_0$ and $B\le 1/\mu'$.
Together with the previous bound, it follows  that $(\pE \alpha_0^4)^{1/4}\ge 1-(\mu_0/\mu')^{1/4}$.
Since $\alpha_0^2\preccurlyeq \lVert  u \rVert_2^2$ and $\{u\}$ satisfies $\lVert  u \rVert_2^2=1$, we have $\pE \alpha_0^2 \ge \pE \alpha_0^4\ge 1-4(\mu_0/\mu')^{1/4}$.
Finally, using $\lVert  u-\cramped{x^{\scriptscriptstyle (0)}} \rVert_2^2=\lVert  u \rVert_2^2-\alpha_0^2$, we derive the desired bound $\pE \lVert  u-\cramped{x^{\scriptscriptstyle (0)}} \rVert_2^2=1-\pE\alpha_0^2\le 4(\mu_0/\mu')^{1/4}$
thus establishing Lemma~\ref{lem:pseudo-hoelder} and Theorem~\ref{thm:sparse-recovery}. \qed

\subsection{Sparse dictionary learning}

A \emph{$\kappa$-overcomplete dictionary} is a matrix $A\in \mathbb{R}^{n\times m}$ with $\kappa=m/n\ge 1$ and isotropic unit vectors as columns (so that $\lVert  A^\top u \rVert_2^2=\kappa \lVert  u \rVert_2^2$).
We say a distribution $\{x\}$ over $\mathbb{R}^m$ is $(d,\tau)$-nice if it satisfies $\bbE_i x_i^d=1$ and $\bbE_i x_i^{d/2} x_j^{d/2}\le \tau$ for all $i\neq j\in [m]$, and it satisfies that non-square monomial degree-$d$ moments vanish so that $\bbE x^\alpha =0$ for all non-square degree-$d$ monomials $x^\alpha$,
where $x^\alpha = \prod x_i^{\alpha_i}$ for $\alpha \in \mathbb{Z}^n$.
For $d=O(1)$ and $\tau = o(1)$, a nice distribution satisfies that $\bbE \tfrac{1}{m}\sum_i x_i^4 \gg \left( \tfrac{1}{m}\sum_i x_i^2 \right)^2$ which means that it is approximately sparse in the sense that the square of the entries of $x$ has large variance (which means that few of the entries have very big magnitude compared to the rest).

\begin{theorem}

For every $\varepsilon>0$ and $\kappa\ge 1$, there exists $d$ and $\tau$ and a quasipolynomial-time algorithm algorithm for \textsc{sparse dictionary learning} with the following guarantees:
Suppose the input consists of $n^{O(1)}$ independent samples\footnote{Here, we also make the mild assumption that the degree-$2d$ moments of $x$ are bounded by $n^{O(1)}$.} from a distribution $\{y=Ax\}$ over $\mathbb{R}^n$, where $A\in\mathbb{R}^{n\times m}$ is a $\kappa$-overcomplete dictionary and the distribution $\{x\}$ over $\mathbb{R}^m$ is $(d,\tau)$-nice.
Then, with high probability, the algorithm outputs a set of vectors with Hausdorff distance\footnote{
  The \emph{Hausdorff distance} between two sets of vectors upper bounds the maximum distance of a point in one of the sets to its closest point in the other set.
  Due to the innate symmetry of the sparse dictionary problem (replacing a column $\cramped{a^{\scriptscriptstyle (i)}}$ of $A$ by $-\cramped{a^{\scriptscriptstyle (i)}}$ might not affect the input distribution), we measure the Hausdorff distance after symmetrizing the sets, i.e., replacing the set $S$ by $S\cup -S$.
}
at most $\varepsilon$ from the set of columns of $A$.

\end{theorem}

\paragraph{Encoding as a system of polynomial equations.}
Let $\cramped{y^{\scriptscriptstyle (1)}},\ldots,\cramped{y^{\scriptscriptstyle (R)}}$ be independent samples from the distribution $\{y=Ax\}$.
Then, we consider the polynomial $P=\tfrac 1 R \sum_i \langle\cramped{y^{\scriptscriptstyle (i)}}, u \rangle^d\in\mathbb{R}[u]_d$.
Using the properties of nice distributions, a direct computation shows that with high probability $P$ satisfies the relation
\begin{displaymath}
  \lVert  A^\top u \rVert_d^d -\tau \lVert  u \rVert_2^d\preccurlyeq P\preccurlyeq \lVert  A^\top u\rVert_d^d+\tau\lVert  u \rVert_2^d\,.
\end{displaymath}
(Here, we are omitting some constant factors, depending on $d$, that are not important for the following discussion.)
It follows that $P(\cramped{a^{\scriptscriptstyle (i)}})=1\pm \tau$ for every column $\cramped{a^{\scriptscriptstyle (i)}}$ of $A$.
It's also not hard to show that every unit vector $a^\ast$ with $P(a^\ast)\approx 1$ is close to one of the columns of $A$.
(Indeed, every unit vector satisfies $P(a^\ast)\le \max_i \langle \cramped{a^{\scriptscriptstyle (i)}},a^\ast \rangle^{d-2}\kappa+\tau$.
Therefore, $P(a^\ast)\approx 1$ implies that $\langle  \cramped{a^{\scriptscriptstyle (i)}},a^\ast \rangle^2 \ge \kappa^{-\Omega(1/d)}$, which is close to $1$ for $d\gg \log \kappa$.)
What we will show is that pseudodistributions of degree $O(\log n)$ allow us to find all such vectors.

\paragraph{Why does the sum-of-squares method work?}

In the following, $\varepsilon>0$ and $\kappa\ge 1$ are arbitrary constants that determine constants $d=d(\varepsilon,\kappa)\ge 1$ and $\tau=\tau(\varepsilon,\kappa)>0$ (as in the theorem).

\begin{lemma}
\label{lem:isolate}
  Let $P\in\mathbb{R}[u]$ be a degree-$d$ polynomial with $\pm(P-\lVert  A^\top u \rVert_d^d)\preccurlyeq \tau \lVert  u \rVert_2^d$ for some $\kappa$-overcomplete dictionary $A$.
  Let $\mathcal{D}$ be a degree-$O(\log n)$ pseudodistribution that satisfies the constraints $\{ \lVert  u \rVert_2^2=1 \}$ and $\{ P(u)=1-\tau \}$.
  Let $W\in\mathbb{R}[u]$ be a product of $O(\log n)$ random linear forms\footnote{Here, a random linear form means a polynomial $\langle  u,v \rangle\in\mathbb{R}[u]$ where $v$ is a random unit vector in $\mathbb{R}^n$.}.
  Then, with probability at least $n^{-O(1)}$ over the choice of $W$, there exists a column $\cramped{a^{\scriptscriptstyle (i)}}$ of $A$ such that
  \begin{displaymath}
    \tfrac 1 {\pE_\mathcal{D} W^2} \pE_\mathcal{D} W^2 \cdot \left( \lVert  u\rVert^2 - \langle  \cramped{a^{\scriptscriptstyle (i)}},u  \rangle^2 \right) \le \varepsilon
    \,.
  \end{displaymath}
\end{lemma}

\noindent
If $\pE_\mathcal{D}$ is a pseudoexpectation operator, then $\pE_{\mathcal{D}'}\colon P\mapsto \pE W^2 P / \pE W^2$ is also a pseudoexpectation operator (as it satisfies linearity, normalization, and nonnegativity).
(This transformation corresponds to reweighing the pseudodistribution $\mathcal{D}$ by the polynomial $W^2$.)
Hence, the conclusion of the lemma gives us a new pseudodistribution $\mathcal{D}'$ such that $\pE_{\mathcal{D}'} \lVert  u \rVert_2^2 - \langle\cramped{a^{\scriptscriptstyle (i)}},u\rangle^2 \le \varepsilon$.
Therefore, if we sample a Gaussian vector $a^\ast$ with the same quadratic moments as $\mathcal{D}'$, it satisfies $\lVert  a^\ast \rVert_2^2-\langle \cramped{a^{\scriptscriptstyle (i)}}, a^*\rangle^2\le 4 \varepsilon$ with probability at least $3/4$.
At the same time, it satisfies $\lVert  a^\ast \rVert^2\ge 1/4$ with probability at least $1/2$.
Taking these bounds together, $a^\ast$ satisfies $\langle  \cramped{a^{\scriptscriptstyle (i)}},a^\ast \rangle^2\ge (1-16 \varepsilon)\lVert  a^\ast \rVert^2$ with probability at least $1/4$.

Lemma~\ref{lem:isolate} allows us to reconstruct one of the columns of $A$.
Using similar ideas, we can iterate this argument and recover one-by-one all columns of $A$.
We omit the proof of Lemma~\ref{lem:isolate}, but the idea behind it is to first give an SOS proof version of our argument above that maximizers of $P$ must be close to one of the $a^{(i)}$'s.
We then note that if a distribution $\mathcal{D}$ is supported  (up to noise) on at most $m$ different vectors, then we can essentially isolate one of these vectors by re-weighing $\mathcal{D}$ using the product of the squares of $O(\log m)$ random linear forms.
  
It turns out, this latter argument has a low degree SOS proof as well, which means that in our case that given  $\mathcal{D}$ satisfying the constraint $\{ P(u)=1-\tau \}$, we can isolate one of the $\cramped{a^{\scriptscriptstyle (i)}}$'s even when $\mathcal{D}$ is not an actual distribution but merely a pseudodistribution.


\section{Hypercontractive norms and small-set expansion}
\label{sec:SSEvsSOS}

So far we have discussed the Small-Set Expansion Hypothesis and the Sum of Squares algorithm.
We now discuss how these two notions are related.
One connection, mentioned before, is that the SSEH predicts that in many settings the guarantees of the degree-$2$ SOS algorithm are best possible, and so in particular it means that going from degree $2$ to say degree $100$ should not give any substantial improvement in terms of guarantees.
Another, perhaps more meaningful connection is that there is a candidate approach for refuting the SSEH using the SOS algorithm.
At the heart of this approach is the following observation:

\begin{quote}
\emph{The small-set expansion problem is a special case of the problem of finding ``sparse'' vectors in a linear subspace.}
\end{quote}

This may seem strange, as a priori, the following two problem  seem completely unrelated:
\textbf{(i)} Given a graph $G=(V,E)$, find a ``small'' subset $S\subseteq V$ with low expansion $\phi_G(S)$,
and \textbf{(ii)} Given a subspace $W\subseteq \mathbb{R}^n$, find a ``sparse'' vector in $W$.
The former is a combinatorial problem on graphs, and the latter a geometric problem on subspaces.
However, for the right notions of ``small'' and ``sparse'', these turn out to be essentially the same problem.
Intuitively, the reason is the following:
the expansion of a set $S$ is proportional to the quantity $x^\top L x$ where $x$ is the characteristic vector of $S$ (i.e. $x_i$ equals $1$ if $i\in S$ and equals $0$ otherwise),
and $L$ is the \emph{Laplacian matrix} of $G$ (defined as $L=I-d^{-1}A$ where $I$ is the identity, $d$ is the degree, and $A$ is $G$'s adjacency matrix).
Let $v_1,\ldots,v_n$ be the eigenvectors of $L$ and $\lambda_1 , \ldots, \lambda_n$ the corresponding eigenvalues.
Then $x^\top L x = \sum_{i=1}^n \lambda_i \langle v_i,x \rangle^2.$

Therefore if $x^\top L x$ is smaller than $\varphi\|x\|^2$ and $c$ is a large enough constant, then  most of the mass of $x$ is contained
in the subspace $W = \mathrm{Span}\{ v_i : \lambda_i \leq c\varphi \}$.
Since $S$ is small, $x$ is sparse, and so we see that there is a sparse vector that is ``almost'' contained in $W$.
Moreover, by projecting $x$ into $W$ we can also find a ``sparse'' vector that is actually contained in $W$,
if we allow a slightly softer notion of ``sparseness'', that instead of stipulating that most coordinates are zero, only requires that the distribution
of coordinates is very ``spiky'' in the sense that most of its mass is dominated by the few ``heavy hitters''.

Concretely, for $p>1$ and $\delta\in(0,1)$, we say that a vector $x\in\mathbb{R}^n$ is \emph{$(\delta,p)$-sparse} if $\bbE_i x_i^{2p} \geq \delta^{1-p}(\bbE_i x_i^2)^p$.
Note that a characteristic vector of a set of measure $\delta$ is $(\delta,p)$-sparse for any $p$.
The relation between small-set-expansion and finding sparse vectors in a subspace is captured by the following theorem:

\begin{theorem}[Hypercontractivity and small-set expansion~\cite{BarakBHKSZ12}, informal statement]\label{thm:SSEvsNorms}
Let $G=(V,E)$ be a $d$-regular graph with Laplacian $L$. Then for every  $p\geq 2$ and $\varphi \in (0,1)$,
\begin{enumerate}
\item \emph{(Non-expanding small sets imply sparse vectors.)} If there exists  $S\subseteq V$ with $|S|=o(|V|)$  and  $\phi_G(S) \leq \varphi$
then there exists an $(o(1),p)$-sparse vector  $x\in W_{\leq \varphi + o(1)}$ where for every $\lambda$, $W_{\leq \lambda}$ denotes the span of the eigenvectors of $L$
with eigenvalue smaller than $\lambda$.

\item \emph{(Sparse vectors imply non-expanding small sets.)}
If there exists a $(o(1),p)$-sparse vector $x\in W_{\leq \varphi}$, then there exists $S\subseteq V$ with $|S| = o(|V|)$
and $\phi_G(S) \leq \rho$ for some constant $\rho<1$ depending on $\varphi$.
\end{enumerate}
\end{theorem}

The first direction of Theorem~\ref{thm:SSEvsNorms}  follows from the above reasoning, and was known before the work of~\cite{BarakBHKSZ12}.
The second direction is harder, and we omit the proof here.
The theorem reduces the question of determining whether there for small sets $S$, the minimum of $\phi_G(S)$ is close to one or close to zero,
into the question of bounding the maximum of $\bbE_i x_i^{2p}$ over all unit vectors in some subspace.
The latter question is a polynomial optimization problem of the type the SOS algorithm is designed for!
Thus, we see that we could potentially resolve the SSEH if we could answer the following question:
\begin{quote}
\emph{What is the degree of SOS proofs needed to certify that the $2p$-norm is bounded for all (Euclidean norm) unit vectors in some subspace $W$?}
\end{quote}

We still don't know the answer to this question in full generality, but we do have some interesting special cases.
Lemma~\ref{lem:random-norm-bound} of Section~\ref{sec:sparse-vector-recovery} implies that if $W$ is a random subspace of dimension $\ll \sqrt{n}$ then we can certify
that $\bbE_i x_i^4 \leq O(\bbE_i x_i^2)^2$ for all $x\in W$ via a degree-$4$ SOS proof.
This is optimal, as the $4$-norm simply won't be bounded for dimensions larger than $\sqrt{n}$:

\begin{lemma} \label{lem:dimbound} Let $W \subseteq \mathbb{R}^n$ have dimension $d$ and $p\geq 2$, then there exists a unit vector $x\in W$ such that
\[
\bbE_i x_i^{2p} \geq \tfrac{d^{p}}{n}(\bbE_i x_i^2)^p
\]
\end{lemma}

Hence in particular any subspace of dimension $d \gg n^{1/p}$ contains a $(o(1),p)$-sparse vector.
  
\begin{proof}[Proof of Lemma~\ref{lem:dimbound}]
Let $P$ be the matrix corresponding to the projection operator to the subspace $W$.
Note that $P$ has $d$ eigenvalues equalling $1$ and the rest equal $0$,
and hence $\mathrm{Tr}(P)=d$ and the Frobenius norm squared of $P$, defined as $\sum P_{i,j}^2$, also equals $d$.
Let $x^i = Pe^i$ where $e^i$ is the $i^{th}$ standard basis vector.
Then $\sum x^i_i$ is the trace of $P$ which equals $d$ and hence using
Cauchy-Schwarz
\[
\sum (x^i_i)^2 \geq \frac{1}{n}\left(\sum x_i \right)^2 = \frac{\mathrm{Tr}(P)^2}{n}  = \frac{d^2}{n} \;.
\]
On the other hand,
\[
\sum_i \sum_j (x^i)_j^2  = \sum_{i,j} (Pe^i)_j^2 = \sum P_{i,j}^2 = d \;.
\]
Therefore, by the inequality $(\sum a_i)/(\sum b_i) \leq \max a_i/b_i$, there exists an $i$ such that
if we let $x =x^i$ then
\begin{math}
  x_i^2 \geq \tfrac{d}{n}\sum_j x_j^2 = d \bbE x_j^2.
\end{math}
Hence, just the contribution of the $i^{th}$ coordinate to the expectation achieves
\begin{math}
  \bbE_j x_j^{2p} \geq \tfrac{d^{p}}{n} \left( \bbE_j x_j^2 \right)^p
  .
\end{math}
\end{proof}

Lemma~\ref{lem:dimbound} implies the following corollary:

\begin{corollary}\label{cor:dimbound} Let $p,n\in \mathbb{N}$, and $W$ be subspace of $\mathbb{R}^n$.
If $\bbE_i x_i^{2p} \leq O(\bbE_i x_i^2)^p)$,
then there is an $O(n^{1/p})$-degree SOS proof for
this fact.
(The constants in the $O(\cdot)$ notation can depend on $p$ but not on $n$.)
\end{corollary}
\begin{proof}[Proof sketch]
By Lemma~\ref{lem:dimbound}, the condition implies that $d=\dim W \leq O(n^{1/p})$,
and it is known that approximately bounding a degree-$O(1)$ polynomial on the $d$-dimensional sphere
requires an SOS proof of at most $O(d)$ degree (e.g., see~\cite{doherty2012convergence} and the references therein).
\end{proof}

Combining Corollary~\ref{cor:dimbound} with Theorem~\ref{thm:SSEvsNorms} implies that for every $\varepsilon,\delta$ there exists some $\tau$ (tending to zero
with $\varepsilon$),  such that if we want to distinguish between the case that an $n$-vertex graph $G$ satisfies $\phi_G(S) \leq \varepsilon$ for every $|S|\leq \delta n$,
and the case that there exists some $S$ of size at most $\delta n$ with $\phi_G(S) \geq 1-\varepsilon$, then we can do so using
a degree $n^{\tau}$ SOS proofs, and hence in $\exp(O(n^{\tau}))$ time.
This is much better than the trivial $\binom{n}{\delta n}$ time algorithm that enumerates all possible sets.
Similar ideas can be used to achieve an algorithm with a similar running time for the problem underlying the Unique Games Conjecture~\cite{AroraBS10}.
If these algorithms could be improved so the exponent $\tau$ tends to zero with $n$ for a fixed $\varepsilon$, this would essentially refute the
SSEH and UGC.

Thus, the question is whether Corollary~\ref{cor:dimbound} is the best we could do.
As we've seen, Lemma~\ref{lem:random-norm-bound} shows that for random subspaces we can do much better, namely certify the bound with a constant degree proof.
Two other results are known of that flavor.
Barak, Kelner and Steurer~\cite{BarakKS14} showed that if a $d$-dimensional subspace $W$ does not contain
a $(\delta,2)$-sparse vector, then there is an $O(1)$-degree SOS proof that it does not contain (or even almost contains) a vector with $O(\tfrac{\delta n}{d^{1/3}})$ nonzero coordinates.
If the dependence on $d$ could be eliminated (even at a significant cost to the degree), then this would also refute the SSEH.
Barak, Brand{\~a}o, Harrow, Kelner, Steurer and Zhou~\cite{BarakBHKSZ12} gave an $O(1)$-degree SOS proof
for the so-called ``Bonami-Beckner-Gross $(2,4)$ hypercontractivity theorem`` (see~\cite[Chap.~9]{Odonnell14}). This is the statement that for every constant $k$, the subspace $W_k \subseteq \mathbb{R}^{2^t}$ containing the evaluations of all degree $\leq k$ polynomials on the points $\{ \pm 1 \}^t$ does not contain an $(o(1),2)$-sparse vector, and specifically satisfies for all $x\in W_k$,
\begin{equation}
\label{eq:hyper-contract}
\bbE x_i^4 \leq 9^k ( \bbE x_i^2 )^2
\,.
\end{equation}
On its own this might not seem so impressive, as this is just one particular subspace.
However, this particular subspace underlies much of the evidence that has been offered so far in support of both the UGC and SSEH conjectures.
The main evidence for the UGC/SSEH consists of several papers such as~\cite{KhotV05,KhotS09,RaghavendraS09,BarakGHMRS12} that verified the predictions of these conjectures by proving that various natural algorithms indeed fail to solve some of the computational problems that are hard if the conjectures are true.
These results all have the form of coming up with a ``hard instance'' $G$ on which some algorithm $\mathcal{A}$ fails, and so to prove such a result one needs to do two things: \textbf{(i)} compute (or bound) the true value of the parameter on $G$, and \textbf{(ii)} show that the value that $\mathcal{A}$ outputs on $G$ is (sufficiently) different than this true value.
It turns out that all of these papers, the proof of \textbf{(i)} can be formulated as low degree SOS proof, and in fact the heart of these proofs is the bound (\ref{eq:hyper-contract}).
Therefore, the results of~\cite{BarakBHKSZ12} showed that all these  ``hard instances'' can in fact be solved by the SOS algorithm using a constant degree.
This means that at the moment, we don't even have any example of an instance for the problems underlying the SSEH and UGC that can be reasonably \emph{conjectured} (let alone proved) hard for the constant degree SOS algorithm.
This does not mean that such instances do not exist, but is suggestive that we have not yet seen the last algorithmic word on this question.


\paragraph{Acknowledgments.}
We thank Amir Ali Ahmadi for providing us with references on the diverse applications of the SOS method.


\begin{thebibliography}{\biblkeystyle{BGH{\etalchar{+}}12}}
\biblbegin

\bibitem[\biblkeystyle{AAJ{\etalchar{+}}13}]{AgarwalANT13}
A.~Agarwal, A.~Anandkumar, P.~Jain, P.~Netrapalli, and R.~Tandon.
\newblock Learning Sparsely Used Overcomplete Dictionaries via Alternating
  Minimization.
\newblock {\em arXiv preprint 1310.7991}, 2013.

\bibitem[\biblkeystyle{Alo86}]{Alon86}
N.~Alon.
\newblock Decomposition of the complete{\it r}-graph into complete{\it
  r}-partite{\it r}-graphs.
\newblock {\em Graphs and Combinatorics}, 2(1):95--100, 1986.

\bibitem[\biblkeystyle{AM85}]{AlonM85}
N.~Alon and V.~D. Milman.
\newblock $\lambda_{\mbox{1}}$, Isoperimetric inequalities for graphs, and
  superconcentrators.
\newblock {\em J. Comb. Theory, Ser. B}, 38(1):73--88, 1985.

\bibitem[\biblkeystyle{AMS11}]{AmbuhlMS11}
C.~Amb{\"u}hl, M.~Mastrolilli, and O.~Svensson.
\newblock Inapproximability Results for Maximum Edge Biclique, Minimum Linear
  Arrangement, and Sparsest Cut.
\newblock {\em SIAM J. Comput.}, 40(2):567--596, 2011.

\bibitem[\biblkeystyle{ABS10}]{AroraBS10}
S.~Arora, B.~Barak, and D.~Steurer.
\newblock Subexponential Algorithms for Unique Games and Related Problems.
\newblock In {\em FOCS}, pages 563--572, 2010.

\bibitem[\biblkeystyle{AGM13}]{AroraGM13}
S.~Arora, R.~Ge, and A.~Moitra.
\newblock New Algorithms for Learning Incoherent and Overcomplete Dictionaries.
\newblock {\em arXiv preprint 1308.6723}, 2013.

\bibitem[\biblkeystyle{ARV09}]{AroraRV09}
S.~Arora, S.~Rao, and U.~V. Vazirani.
\newblock Expander flows, geometric embeddings and graph partitioning.
\newblock {\em J. ACM}, 56(2), 2009.

\bibitem[\biblkeystyle{Art27}]{Artin27}
E.~Artin.
\newblock {\"U}ber die zerlegung definiter funktionen in quadrate.
\newblock In {\em Abhandlungen aus dem mathematischen Seminar der
  Universit{\"a}t Hamburg}, volume~5, pages 100--115. Springer, 1927.

\bibitem[\biblkeystyle{Bar12}]{Barak12}
B.~Barak.
\newblock Truth vs. Proof in Computational Complexity.
\newblock {\em Bulletin of the European Association for Theoretical Computer
  Science}, (108), October 2012.

\bibitem[\biblkeystyle{Bar14a}]{Barak14}
B.~Barak.
\newblock Fun and Games with Sums of Squares, Feb. 2014.
\newblock Windows on Theory blog,  \url{http://windowsontheory.org}

\bibitem[\biblkeystyle{Bar14b}]{Barak13}
B.~Barak.
\newblock Structure vs. Combinatorics in Computational Complexity.
\newblock {\em Bulletin of the EATCS}, (112):115--126, February 2014.

\bibitem[\biblkeystyle{BBH{\etalchar{+}}12}]{BarakBHKSZ12}
B.~Barak, F.~G. S.~L. Brand{\~a}o, A.~W. Harrow, J.~A. Kelner, D.~Steurer, and
  Y.~Zhou.
\newblock Hypercontractivity, sum-of-squares proofs, and their applications.
\newblock In {\em STOC}, pages 307--326, 2012.

\bibitem[\biblkeystyle{BGH{\etalchar{+}}12}]{BarakGHMRS12}
B.~Barak, P.~Gopalan, J.~H{\aa}stad, R.~Meka, P.~Raghavendra, and D.~Steurer.
\newblock Making the Long Code Shorter.
\newblock In {\em FOCS}, pages 370--379, 2012.

\bibitem[\biblkeystyle{BKS14a}]{BarakKS15}
B.~Barak, J.~Kelner, and D.~Steurer.
\newblock Dictionary Learning via the Sum-of-Squares Method.
\newblock Unpublished manuscript, 2014.

\bibitem[\biblkeystyle{BKS14b}]{BarakKS14}
B.~Barak, J.~Kelner, and D.~Steurer.
\newblock Rounding Sum of Squares Relaxations.
\newblock In {\em STOC}, 2014.

\bibitem[\biblkeystyle{BPT13}]{blekherman2013semidefinite}
G.~Blekherman, P.~A. Parrilo, and R.~R. Thomas.
\newblock {\em Semidefinite optimization and convex algebraic geometry},
  volume~13.
\newblock Siam, 2013.

\bibitem[\biblkeystyle{Che70}]{cheeger1970lower}
J.~Cheeger.
\newblock A lower bound for the smallest eigenvalue of the Laplacian.
\newblock {\em Problems in analysis}, 625:195--199, 1970.

\bibitem[\biblkeystyle{CHKX06}]{chen2006strong}
J.~Chen, X.~Huang, I.~A. Kanj, and G.~Xia.
\newblock Strong computational lower bounds via parameterized complexity.
\newblock {\em Journal of Computer and System Sciences}, 72(8):1346--1367,
  2006.

\bibitem[\biblkeystyle{Com94}]{comon1994independent}
P.~Comon.
\newblock Independent component analysis, a new concept?
\newblock {\em Signal processing}, 36(3):287--314, 1994.

\bibitem[\biblkeystyle{DH13}]{DemanetH13}
L.~{Demanet} and P.~{Hand}.
\newblock {Recovering the Sparsest Element in a Subspace}, Oct. 2013.
\newblock Arxiv preprint 1310.1654.

\bibitem[\biblkeystyle{Dod84}]{dodziuk1984difference}
J.~Dodziuk.
\newblock Difference equations, isoperimetric inequality and transience of
  certain random walks.
\newblock {\em Transactions of the American Mathematical Society},
  284(2):787--794, 1984.

\bibitem[\biblkeystyle{DW12}]{doherty2012convergence}
A.~C. Doherty and S.~Wehner.
\newblock Convergence of SDP hierarchies for polynomial optimization on the
  hypersphere.
\newblock {\em arXiv preprint arXiv:1210.5048}, 2012.

\bibitem[\biblkeystyle{DF95}]{downey1995fixed}
R.~G. Downey and M.~R. Fellows.
\newblock Fixed-parameter tractability and completeness II: On completeness
  for W[1].
\newblock {\em Theoretical Computer Science}, 141(1):109--131, 1995.

\bibitem[\biblkeystyle{GW95}]{GoemansW95}
M.~X. Goemans and D.~P. Williamson.
\newblock Improved Approximation Algorithms for Maximum Cut and Satisfiability
  Problems Using Semidefinite Programming.
\newblock {\em J. ACM}, 42(6):1115--1145, 1995.

\bibitem[\biblkeystyle{GVX14}]{GoyalVX13}
N.~Goyal, S.~Vempala, and Y.~Xiao.
\newblock Fourier PCA.
\newblock In {\em STOC}, 2014.
\newblock Also available as arXiv report 1306.5825.

\bibitem[\biblkeystyle{Gri01}]{Grigoriev01}
D.~Grigoriev.
\newblock Linear lower bound on degrees of Positivstellensatz calculus proofs
  for the parity.
\newblock {\em Theor. Comput. Sci.}, 259(1-2):613--622, 2001.

\bibitem[\biblkeystyle{GV01}]{GrigorievV01}
D.~Grigoriev and N.~Vorobjov.
\newblock Complexity of Null-and Positivstellensatz proofs.
\newblock {\em Annals of Pure and Applied Logic}, 113(1):153--160, 2001.

\bibitem[\biblkeystyle{Gro53}]{grothendieck1953resume}
A.~Grothendieck.
\newblock R{\'e}sum{\'e} de la th{\'e}orie m{\'e}trique des produits tensoriels
  topologiques.
\newblock {\em Bol. Soc. Mat. Sao Paulo}, 8(1-79):88, 1953.

\bibitem[\biblkeystyle{H{\aa}s96}]{Hastad96}
J.~H{\aa}stad.
\newblock Clique is Hard to Approximate Within n$^{1-\epsilon}$.
\newblock In {\em FOCS}, pages 627--636, 1996.

\bibitem[\biblkeystyle{Kho01}]{Khot01}
S.~Khot.
\newblock Improved Inaproximability Results for MaxClique, Chromatic Number and
  Approximate Graph Coloring.
\newblock In {\em FOCS}, pages 600--609, 2001.

\bibitem[\biblkeystyle{Kho02}]{Khot02}
S.~Khot.
\newblock On the Power of Unique 2-Prover 1-Round Games.
\newblock In {\em IEEE Conference on Computational Complexity}, page~25, 2002.

\bibitem[\biblkeystyle{Kho10a}]{Khot10a}
S.~Khot.
\newblock Inapproximability of np-complete problems, discrete fourier analysis,
  and geometry.
\newblock In {\em International Congress of Mathematics}, volume~5, 2010.

\bibitem[\biblkeystyle{Kho10b}]{Khot10b}
S.~Khot.
\newblock On the Unique Games Conjecture (Invited Survey).
\newblock In {\em 2012 IEEE 27th Conference on Computational Complexity}, pages
  99--121. IEEE, 2010.

\bibitem[\biblkeystyle{KS09}]{KhotS09}
S.~Khot and R.~Saket.
\newblock SDP Integrality Gaps with Local $\ell_1$-Embeddability.
\newblock In {\em FOCS}, pages 565--574, 2009.

\bibitem[\biblkeystyle{KV05}]{KhotV05}
S.~Khot and N.~K. Vishnoi.
\newblock The Unique Games Conjecture, Integrality Gap for Cut Problems and
  Embeddability of Negative Type Metrics into {$\ell_1$}.
\newblock In {\em FOCS}, pages 53--62, 2005.

\bibitem[\biblkeystyle{Kri64}]{Krivine64}
J.-L. Krivine.
\newblock Anneaux pr{\'e}ordonn{\'e}s.
\newblock {\em Journal d'analyse math{\'e}matique}, 12(1):307--326, 1964.

\bibitem[\biblkeystyle{Las01}]{Lasserre01}
J.~B. Lasserre.
\newblock Global Optimization with Polynomials and the Problem of Moments.
\newblock {\em SIAM Journal on Optimization}, 11(3):796--817, 2001.

\bibitem[\biblkeystyle{Lau03}]{Laurent03}
M.~Laurent.
\newblock A Comparison of the Sherali-Adams, Lov{\'a}sz-Schrijver, and Lasserre
  Relaxations for 0-1 Programming.
\newblock {\em Math. Oper. Res.}, 28(3):470--496, 2003.

\bibitem[\biblkeystyle{Lau09}]{Laurent09}
M.~Laurent.
\newblock Sums of squares, moment matrices and optimization over polynomials.
\newblock In {\em Emerging applications of algebraic geometry}, pages 157--270.
  Springer, 2009.

\bibitem[\biblkeystyle{Lov79}]{lovasz1979shannon}
L.~Lov{\'a}sz.
\newblock On the Shannon capacity of a graph.
\newblock {\em Information Theory, IEEE Transactions on}, 25(1):1--7, 1979.

\bibitem[\biblkeystyle{LS91}]{LovaszS91}
L.~Lov{\'a}sz and A.~Schrijver.
\newblock Cones of matrices and set-functions and 0-1 optimization.
\newblock {\em SIAM Journal on Optimization}, 1(2):166--190, 1991.

\bibitem[\biblkeystyle{Nes00}]{Nesterov00}
Y.~Nesterov.
\newblock Squared functional systems and optimization problems.
\newblock {\em High performance optimization}, 13:405--440, 2000.

\bibitem[\biblkeystyle{O'D14}]{Odonnell14}
R.~O'Donnell.
\newblock {\em Analysis of Boolean Functions}.
\newblock Cambridge University Press, 2014.
\newblock To be published in May 2014.

\bibitem[\biblkeystyle{OF96}]{olshausen1996emergence}
B.~A. Olshausen and D.~J. Field.
\newblock Emergence of simple-cell receptive field properties by learning a
  sparse code for natural images.
\newblock {\em Nature}, 381(6583):607--609, 1996.

\bibitem[\biblkeystyle{Par00}]{Parrilo00}
P.~A. Parrilo.
\newblock {\em Structured semidefinite programs and semialgebraic geometry
  methods in robustness and optimization}.
\newblock PhD thesis, California Institute of Technology, 2000.

\bibitem[\biblkeystyle{Rag08}]{Raghavendra08}
P.~Raghavendra.
\newblock Optimal algorithms and inapproximability results for every CSP?
\newblock In {\em STOC}, pages 245--254, 2008.

\bibitem[\biblkeystyle{RS09a}]{RaghavendraS09}
P.~Raghavendra and D.~Steurer.
\newblock Integrality Gaps for Strong SDP Relaxations of UNIQUE GAMES.
\newblock In {\em FOCS}, pages 575--585, 2009.

\bibitem[\biblkeystyle{RS09b}]{RaghavendraS09groth}
P.~Raghavendra and D.~Steurer.
\newblock Towards computing the Grothendieck constant.
\newblock In {\em SODA}, pages 525--534, 2009.

\bibitem[\biblkeystyle{RS10}]{RaghavendraS10}
P.~Raghavendra and D.~Steurer.
\newblock Graph expansion and the unique games conjecture.
\newblock In {\em STOC}, pages 755--764, 2010.

\bibitem[\biblkeystyle{RST10}]{RaghavendraST10}
P.~Raghavendra, D.~Steurer, and P.~Tetali.
\newblock Approximations for the isoperimetric and spectral profile of graphs
  and related parameters.
\newblock In {\em STOC}, pages 631--640, 2010.

\bibitem[\biblkeystyle{RST12}]{RaghavendraST12}
P.~Raghavendra, D.~Steurer, and M.~Tulsiani.
\newblock Reductions between Expansion Problems.
\newblock In {\em IEEE Conference on Computational Complexity}, pages 64--73,
  2012.

\bibitem[\biblkeystyle{Rez00}]{reznick2000some}
B.~Reznick.
\newblock Some concrete aspects of Hilbert's 17th problem.
\newblock {\em Contemporary Mathematics}, 253:251--272, 2000.

\bibitem[\biblkeystyle{Sch08}]{Schoenebeck08}
G.~Schoenebeck.
\newblock Linear Level {Lasserre} Lower Bounds for Certain k-{CSPs}.
\newblock In {\em FOCS}, pages 593--602, 2008.

\bibitem[\biblkeystyle{SA90}]{SheraliA90}
H.~D. Sherali and W.~P. Adams.
\newblock A hierarchy of relaxations between the continuous and convex hull
  representations for zero-one programming problems.
\newblock {\em SIAM Journal on Discrete Mathematics}, 3(3):411--430, 1990.

\bibitem[\biblkeystyle{Sho87}]{Shor87}
N.~Shor.
\newblock An approach to obtaining global extremums in polynomial mathematical
  programming problems.
\newblock {\em Cybernetics and Systems Analysis}, 23(5):695--700, 1987.

\bibitem[\biblkeystyle{SWW12}]{SpielmanWW12}
D.~A. Spielman, H.~Wang, and J.~Wright.
\newblock Exact Recovery of Sparsely-Used Dictionaries.
\newblock {\em Journal of Machine Learning Research - Proceedings Track},
  23:37.1--37.18, 2012.

\bibitem[\biblkeystyle{Ste74}]{Stengle74}
G.~Stengle.
\newblock A Nullstellensatz and a Positivstellensatz in semialgebraic geometry.
\newblock {\em Mathematische Annalen}, 207(2):87--97, 1974.

\bibitem[\biblkeystyle{Ste10}]{Steurer10}
D.~Steurer.
\newblock Fast {SDP} Algorithms for Constraint Satisfaction Problems.
\newblock In {\em SODA}, pages 684--697, 2010.

\bibitem[\biblkeystyle{Tre12}]{Trevisan12}
L.~Trevisan.
\newblock On Khot's Unique Games Conjecture.
\newblock {\em Bulletin (New Series) of the American Mathematical Society},
  49(1), 2012.

\bibitem[\biblkeystyle{Tul09}]{Tulsiani09}
M.~Tulsiani.
\newblock CSP gaps and reductions in the lasserre hierarchy.
\newblock In {\em STOC}, pages 303--312, 2009.

\biblend
\end{thebibliography}

\newcommand{\etalchar}[1]{$^{#1}$}
\def\cprime{$'$} \def\cprime{$'$} \def\cprime{$'$} \def\cprime{$'$}
  \def\cprime{$'$} \def\cprime{$'$} \def\cprime{$'$} \def\cprime{$'$}
  \def\cprime{$'$} \def\cprime{$'$} \def\cprime{$'$}
  \def\polhk#1{\setbox0=\hbox{#1}{\ooalign{\hidewidth
  \lower1.5ex\hbox{`}\hidewidth\crcr\unhbox0}}} \def\cprime{$'$}
  \def\cprime{$'$} \def\cprime{$'$} \def\cprime{$'$} \def\cprime{$'$}
  \def\cprime{$'$} \def\cprime{$'$} \def\cprime{$'$} \def\cprime{$'$}
  \def\cprime{$'$} \def\cprime{$'$}
  \def\cfac#1{\ifmmode\setbox7\hbox{$\accent"5E#1$}\else
  \setbox7\hbox{\accent"5E#1}\penalty 10000\relax\fi\raise 1\ht7
  \hbox{\lower1.15ex\hbox to 1\wd7{\hss\accent"13\hss}}\penalty 10000
  \hskip-1\wd7\penalty 10000\box7} \def\cprime{$'$} \def\cprime{$'$}
  \def\cprime{$'$} \def\cprime{$'$} \def\cprime{$'$} \def\cprime{$'$}
  \def\ocirc#1{\ifmmode\setbox0=\hbox{$#1$}\dimen0=\ht0 \advance\dimen0
  by1pt\rlap{\hbox to\wd0{\hss\raise\dimen0
  \hbox{\hskip.2em$\scriptscriptstyle\circ$}\hss}}#1\else {\accent"17 #1}\fi}
\makeatletter
\providecommand{\biblkeystyle}[1]{#1}
\providecommand{\biblbegin}{}
\providecommand{\biblend}{}
\@ifundefined{beginbibabs}{\def\beginbibabs{\begin{quotation}\noindent}
\def\endbibabs{\end{quotation}}}{}
\makeatother

\end{document}